\documentclass[a4paper,12pt]{article}

\usepackage[english]{babel}
\usepackage{hyperref}
\usepackage{endnotes}
\usepackage{amsmath}
\usepackage{amsfonts}
\usepackage{collect}
\usepackage[all,cmtip]{xy}
\usepackage{amsmath, amsfonts,xfrac, amsthm,amssymb,fancyhdr}

\usepackage{mathtools}

\usepackage{setspace}

\usepackage{geometry}
\usepackage{xypic,colortbl}

\usepackage{fixltx2e, ifthen, sidecap, wrapfig, xspace}

\usepackage{mathrsfs, stmaryrd}

\usepackage{sectsty}

\usepackage{
 calc, enumitem, makeidx,bbold
}

\usepackage[font=small,format=plain,labelfont=sc,textfont=it]{caption}
\usepackage[T1]{fontenc}
\usepackage[utf8]{inputenc}
\usepackage[all]{xy}
\usepackage[geometry]{ifsym}
\usepackage[cmyk]{xcolor}

\newcommand{\mathsym}[1]{{}}

{\begin{figure}[#1]
\centering 
\includegraphics[scale=#3]{#2}
}
{\end{figure}}

\newlist{enumeratep}{enumerate}{10}
\setlist[enumeratep]{label=\quad\textit{\arabic*'.},ref=\arabic*',leftmargin=*}

\newenvironment{romanlist'}[0]
{\begin{list}{\makebox[0.5cm][l]{\textit{\roman{enumi}')}}}{\usecounter{enumi}}}
{\end{list}}

\newcommand{\savelabel}[2]{\expandafter\newtoks\csname#1\endcsname
  \global\csname#1\endcsname={#2} \label{#1} #2}
\newcommand{\loadlabel}[1]{\noindent {\bf Lemma~\ref{#1}. } \textit{\the\csname#1\endcsname}
\medskip

}
\renewcommand{\setminus}{-}

\newcommand{\loadlabelthm}[1]{\medskip\noindent {\bf Theorem~\ref{#1}. }
  \noindent  \textit{\the\csname#1\endcsname}
\medskip
}

\newcommand{\loadlabelprop}[1]{\noindent {\bf Proposition~\ref{#1}. }
  \noindent  \textit{\the\csname#1\endcsname}
\medskip

}

\renewcommand{\subset}{\subseteq}

\newcommand{\atleast}[1]{{\ge n}}
\newcommand{\less}[1]{{<n}}

\newsavebox{\quoteitbox}

{\hspace*{\fill}{\upshape(\usebox{\quoteitbox})}\end{quote}%
\end{sloppypar}\noindent\ignorespacesafterend}

\newenvironment{quoteit*} 
{\begin{sloppypar}\noindent\slshape\begin{quote}\itshape} 
	{\end{quote}\ignorespaces\end{sloppypar}\noindent\ignorespacesafterend}

\newenvironment{quotetag*}
{~\par
	\begingroup                  
	\begin{equation*}
		 \begin{minipage}[c]{115mm}
			\it\noindent{\par}
}
{
		\end{minipage}
	\end{equation*}
	\endgroup                        
\par
\textnormal
\medskip
}

\newcommand{\Cc}{{\mathcal C}}

\newcommand{\Ee}{{\mathcal E}}

\newcommand\set[1]{\ensuremath{\{#1\}}}

\newtheoremstyle{theoremstyle}
  {3pt}
  {3pt}
  {\itshape}
  {0pt}
  {\bfseries}
  {}
  {4pt}
  {}

\theoremstyle{theoremstyle}
\newtheorem{theorem}{Theorem}[section]
\newtheorem*{theorem*}{Theorem}

\newtheorem{lemma}[theorem]{Lemma}
\newtheorem{proposition}[theorem]{Proposition}

\newtheorem{claim}{Claim}

\newtheoremstyle{remarkstyle}
  {3pt}
  {10pt}
  {}
  {0pt}
  {\itshape}
  {}
  {4pt}
  {\thmname{#1}\thmnumber{ #2}\thmnote{ (#3)}.}

\theoremstyle{remarkstyle}
\newtheorem{example}{Example}[section]

\newtheoremstyle{definitionstyle}
  {3pt}
  {3pt}
  {}
  {0pt}
  {\itshape}
  {}
  {4pt}
  {\thmname{#1}\thmnumber{ #2}\thmnote{ (#3)}.}

\theoremstyle{definitionstyle}

\newlength{\wideaslength}

\renewcommand{\subset}{\subseteq}

\newcounter{quotecount}

\newcommand{\seta}[1]{}

\newcommand{\dom}[1]{\textrm{Dom}(#1)}

\def\lsim{\mathrel{\rlap{\lower4pt\hbox{\hskip1pt$\sim$}}
    \raise1pt\hbox{$<$}}}                

\definecolor{gray1}{rgb}{0.99,0.99,0.99}
\definecolor{gray2}{rgb}{0.97,0.97,0.97}
\definecolor{gray3}{rgb}{0.95,0.95,0.95}
\definecolor{gray4}{rgb}{0.93,0.93,0.93}
\definecolor{gray5}{rgb}{0.91,0.91,0.91}
\definecolor{gray6}{rgb}{0.89,0.89,0.89}
\definecolor{gray7}{rgb}{0.87,0.87,0.87}
\definecolor{gray8}{rgb}{0.85,0.85,0.85}
\definecolor{gray9}{rgb}{0.83,0.83,0.83}
\definecolor{gray10}{rgb}{0.81,0.81,0.81}
\definecolor{gray20}{rgb}{0.71,0.71,0.71}
\definecolor{gray40}{rgb}{0.51,0.51,0.51}

\newcounter{col}

{\noindent\emph{Proof:} see Appendix~\expandafter\ref\expandafter%
		{col:\arabic{col}}.\stepcounter{col}
  \collect{col}%
  {\subsection{#1}%
  \expandafter\label\expandafter{col:\arabic{col}}%
  \stepcounter{col}%
  }%
  {\qed}%
}{\endcollect\qed\smallskip}

{\noindent(see Appendix~\expandafter\ref\expandafter%
		{col:\arabic{col}}.)\stepcounter{col}\collect{col}%
  {\subsection{Proof of #1}%
  \expandafter\label\expandafter{col:\arabic{col}}%
  \stepcounter{col}%
  }%
  {\qed}%
}{\endcollect}

\newcommand{\Nat}{\mathbb N}

\newcommand{\str}[1]{\mathbb{#1}}

\newcommand{\csp}{\mathrm{CSP}}
\newcommand{\nesetril}{Ne\v{s}et\v{r}il\xspace}
\newcommand{\hubicka}{Hubi\v{c}ka\xspace}
\newcommand{\fraisse}{Fra\"{\i}ss\'e\xspace}

\date{}
\newcommand{\calF}{{\mathcal F}}
\newcommand{\calC}{\mathcal C}
\DeclareMathOperator{\forbh}{Forb_\mathit{h}}

\begin{document}
	\definecollection {col}

        \title{{\bf Non-homogenizable classes \\ of finite
            structures}\footnote{A shorter preliminary version of this
            paper was published in Proceedings of 25th EACSL Annual
            Conference on Computer Science Logic (CSL), LIPIcs series,
            pp. 16:1-16:16, 2016.}}

        \author{Albert Atserias \\ Universitat Polit\`ecnica de Catalunya 
        \and Szymon Toru\'nczyk \\ University of Warsaw}

\maketitle
\newcommand{\myparagraph}[1]{\paragraph{#1.}}
\newcommand{\oldnote}[1]{}

\begin{abstract}\noindent
Homogenization is a powerful way of taming a class of finite
structures with several interesting applications in different areas,
from Ramsey theory in combinatorics to constraint satisfaction
problems (CSPs) in computer science, through (finite) model theory.  A
few sufficient conditions for a class of finite structures to allow
homogenization are known, and here we provide a necessary condition.
This lets us show that certain natural classes are not
homogenizable: 1) the class of locally consistent systems of linear
equations over the two-element field or any finite Abelian group, and
2) the class of finite structures that forbid homomorphisms from a
specific MSO-definable class of structures of treewidth two. In
combination with known results, the first example shows that, up to
pp-interpretability, the CSPs that are solvable by local consistency
methods are distinguished from the rest by the fact that their classes
of locally consistent instances are homogenizable.  The second example
shows that, for MSO-definable classes of forbidden patterns, treewidth
one versus two is the dividing line to homogenizability.

\end{abstract}

\section{Introduction}



A relational structure with a countable domain is called homogeneous
if it is highly symmetric: any isomorphism between any two of its
finite induced substructures extends to an automorphism of the whole
structure. In many areas of combinatorics, logic, discrete geometry,
and computer science, homogeneous structures abound, often in the form
of nicely behaved limit objects for classes of finite
structures. Typical examples include the Rado graph $\cal R$,
which can be seen as the limit of the
class of all finite graphs; the linear order of the rational numbers
$\cal Q$, seen as the limit of all finite linear orders; or the
countable Urysohn space $\cal U$, the limit of all rational metric
spaces. The literature on the subject is very extensive; we refer the
reader to \cite{Macpherson2011} for a recent survey.

Homogeneous structures arise as limits of 
well-behaved classes of
finite structures
in a way made precise by \fraisse's theorem,
which describes them combinatorially in a finitary manner.
The theorem states that a homogeneous structure is 
characterized, up to isomorphism, by its {age}, i.e., the class of its finite induced substructures. Moreover, classes of finite structures arising as ages of homogeneous structures are precisely Fra\"{\i}ss\'e classes, i.e., classes closed under taking induced substructures and under amalgamation -- a form of glueing pairs of structures along a common induced substructure
 (see
\cite{Hodges1997} and Section~\ref{sec:prelim} for precise
definitions). 

Thanks to \fraisse's theorem, combinatorial arguments 
involving finite structures can often be replaced by,
or aided by, arguments involving highly symmetric, infinite 
structures. 
In combinatorics, for example, homogeneous structures appear
unavoidably in structural Ramsey theory \cite{Nesetril1996}. At the
intersection between combinatorics and computer science, homogeneous
structures appear in the theory of logical limit laws for various
models of random graphs \cite{KolaitisProemelRothschild1987}.  In
computer science proper, homogeneous structures appear in the theory
of constraint satisfaction problems \cite{Bodirsky2008}, and automata
theory \cite{BojanczykKlinLasota2011}, and verification~\cite{BojanczykSegoufinTorunczyk2013}.

One of the advantages of working  with homogeneous structures, rather than classes of finite structures,
is that their automorphism groups are very rich.
For example, over a finite
relational signature, the homogeneity of the structure immediately
implies that, up to automorphism, 
it has finitely many elements, pairs of elements, 
triples, etc.
In model
theoretic terms, this means that the structure is $\omega$-categorical
by the classical Ryll-Nardzewski theorem, and its first-order theory
admits elimination of quantifiers. In turn, since in any such
structure there are only finitely many first-order definable relations
of each arity, homogeneous structures over finite relational
signatures are, in a strong technical way,  close to being
finite. 

Thus, with Fra\"{\i}ss\'e's theorem in hand and the many applications
of homogeneous structures in mind, it becomes quite important a task
to identify more Fra\"{\i}ss\'e classes. More generally, one would
like to identify classes of finite structures that are perhaps not
Fra\"{\i}ss\'e classes themselves, but 
appear as reducts of some
Fra\"{\i}ss\'e class over a richer yet finite signature. Such classes
of finite structures are called \emph{homogenizable}
\cite{Covington1990}. The point in case is that the lifted
Fra\"{\i}ss\'e class can be thought of as taming its reduct by
providing a homogeneous structure that plays the role of limit object
for it. Many of the application examples mentioned above do actually
go through lifted Fra\"{\i}ss\'e classes and their corresponding
homogeneous Fra\"{\i}ss\'e limits. 

A noticeable amount of work has gone into providing sufficient
conditions for a class of finite structures to be homogenizable.
Instances include the model-theoretic methods of Covington
\cite{Covington1990}, and the combinatorial explicit constructions of
\hubicka and \nesetril \cite{HubickaNesetril2015}. 
Here we provide a combinatorial necessary condition for
homogenizability (Theorem~\ref{thm:necessary} in
Section~\ref{sec:necessary}). This allows us to prove that certain
natural classes of finite structures previously considered in the
literature are not homogenizable. Since it is known that every Ramsey
class is homogenizable \cite{Nesetril2005}, our result may also be
relevent in the context of Ne{\v{s}}et{\v{r}}il's classification
programme of Ramsey classes. See \cite{Nesetril2005} and the
introduction of the recent survey \cite{HubickaNesetril2016} for more
on this.

Our first example of a non-homogenizable class comes from the theory
of constraint satisfaction problems (CSPs). We show that the class of
locally consistent systems of linear equations over the two-element
field is not homogenizable. More generally, the result holds for
systems of equations over any finite Abelian group. This answers a
question first raised by the first author of this paper in
\cite{AtseriasWeyer2009}. Precisely, by a locally consistent system of
equations we mean one whose satisfiability cannot be refuted by the
$(j,k)$-consistency algorithm for small $j$ and $k$, which is a
well-studied heuristic algorithm for solving CSPs. Moreover, in
combination with the resolution of the Bounded Width Conjecture by
Barto and Kozik \cite{BartoKozik2014}, this shows that the
constraint languages whose classes of locally consistent instances are
homogenizable are, up to pp-interpretability, precisely those that are
solvable by local consistency methods. All this is worked out in
Section~\ref{sec:consistent}.

In Section~\ref{sec:forbidden} we give a second example of a
non-homogenizable class that, in this case, is motivated by the works
of Hubi\v{c}ka and Ne\v{s}et\v{r}il \cite{HubickaNesetril2015}, and
 Erd\"{o}s\oldnote{AA: the proper LaTeX accent for Erdos doesn't compile!},
Tardif, and Tardos \cite{ErdosTardifTardos2013}. It was shown in
\cite{HubickaNesetril2015} that every class of finite structures that
is of the form $\forbh({\cal F})$, where~${\cal F}$ is a
\emph{regular} class of connected finite structures, is
homogenizable. In words,~$\forbh({\cal F})$ is the class of finite
structures that do not admit homomorphisms from any structure in
${\cal F}$. The notion of regularity considered in
\cite{HubickaNesetril2015} is closely related to the notion of
regularity in automata theory, and agrees with it on coloured paths
and trees. However, our second example shows that even if ${\cal F}$
is MSO-definable and has maximum treewidth two, the class
$\forbh({\cal F})$ need not be homogenizable. Since MSO-definability
coincides with automatic regularity for coloured paths and trees, this
shows that for MSO-definable classes, treewidth one versus two of the
forbidden structures in ${\cal F}$ is the dividing line to
homogenizability.

\section{Preliminaries} \label{sec:prelim}

\myparagraph{Signatures, structures, reducts, and expansions} A
relational signature $\Sigma$ is a set of relation symbols
$R_1,R_2,\ldots$, each with an associated natural number called its
arity.  In this paper, we consider only finite relational signatures.
A $\Sigma$-structure $\str A = (A; R_1^{\str A}, R_2^{\str A},
\ldots)$ is composed of a set $A$, called its domain, and a relation
$R^{\str A} \subseteq A^{k}$ on $A$ for each $R$ in $\Sigma$, where
$k$ is the arity of $R$. We say that $R^{\str A}$ is the
interpretation of $R$ in $\str A$.  We write $|\str A|$ to denote the
cardinality of the domain of $\str A$.  A $\Sigma$-structure is
sometimes referred to as a structure over the signature $\Sigma$. If
$\Sigma^+$ is a signature that contains $\Sigma$ and $\str A^+$ is a
$\Sigma^+$-structure, then the $\Sigma$-reduct of $\str A^+$ is the
structure $\str A$ obtained from $\str A^+$ by forgetting all
relations from $\Sigma^+-\Sigma$. In this case, we also say that $\str A^+$
is an expansion of $\str A$. Expansions and reducts are also called
lifts and shadows, ~respectively.

\myparagraph{Substructures, homomorphisms, and embeddings} If $\str A$
is a $\Sigma$-structure and $X$ is a subset of its domain $A$, we
write $\str A[X]$ for the substructure of $\str A$ induced by $X$,
that is, the $\Sigma$-structure with domain $X$ in which each relation
symbol $R$ in $\Sigma$ is interpreted by $R^{\str A} \cap X^{k}$,
where $k$ is the arity of $R$.

Let $\str A$ and $\str B$ be structures over the same relational
signature $\Sigma$. Let $A$ and $B$ denote their domains. A
homomorphism from $\str A$ to $\str B$ is a mapping $f : A \rightarrow
B$ for which the inclusion $f(R^{\str A}) \subseteq R^{\str B}$ holds
for every $R$ in $\Sigma$. The homomorphism is strong if in addition
the inclusion $f(A^k \setminus R^{\str A}) \subseteq B^k \setminus
R^{\str B}$ holds for every $R$ in $\Sigma$, where $k$ is the arity of
$R$.  A monomorphism from $\str A$ to $\str B$ is an injective
homomorphism. Whenever $A$ is a subset of $B$ and the inclusion
mapping $A\to B$ is a monomorphism, we say that $\str A$ is a
substructure of $\str B$.  An embedding from $\str A$ to $\str B$ is
an injective strong homomorphism. Whenever $A$ is a subset of $B$ and
the inclusion mapping $A\to B$ is an embedding, we say that $\str A$
is an induced substructure of $\str B$. An isomorphism from $\str A$
to $\str B$ is a surjective embedding.  If there is an isomorphism
from $\str A$ to $\str B$ we say that the two structures are
isomorphic.
If $f : A \rightarrow B$ is a partial mapping
with domain $X \subseteq A$ and image $Y \subseteq B$, we say that $f$
is a partial homomorphism from $\str A$ to $\str B$ if it is a
homomorphism from $\str A[X]$ to $\str B[Y]$. 
We write ${\str B}\choose {\str A}$ to denote the set of
all embeddings from $\str A$ to $\str B$.  
Sometimes we write $f : \str A \to \str B$ to mean that $f$ is a mapping
from the domain of $\str A$ to the domain of $\str B$.

\myparagraph{Amalgamation}
If $\str B$ and $\str C$ are $\Sigma$-structures with domains $B$ and
$C$, we write $\str B \cup \str C$ for their union, i.e.\ the
$\Sigma$-structure with domain $B \cup C$ and relations $R^{\str B
  \cup \str C} = R^{\str B} \cup R^{\str C}$ for every $R$ in
$\Sigma$. Let $f$ and $g$ be embeddings from the same structure
$\str A$ into structures $\str B$ and $\str C$, respectively. 
The
structure $\str D$ is an amalgam of $\str B$ and $\str C$ through $f$
and $g$ if there exist embeddings $f'$ and $g'$ from $\str B$
to $\str D$ and $\str C$ to $\str D$, respectively, such that the diagram in Figure~\ref{fig:amalg} commutes, i.e., 
$f'
\circ f = g' \circ g$.
\begin{figure}[ht]
{
    \begin{center}
    \includegraphics[height=2in]{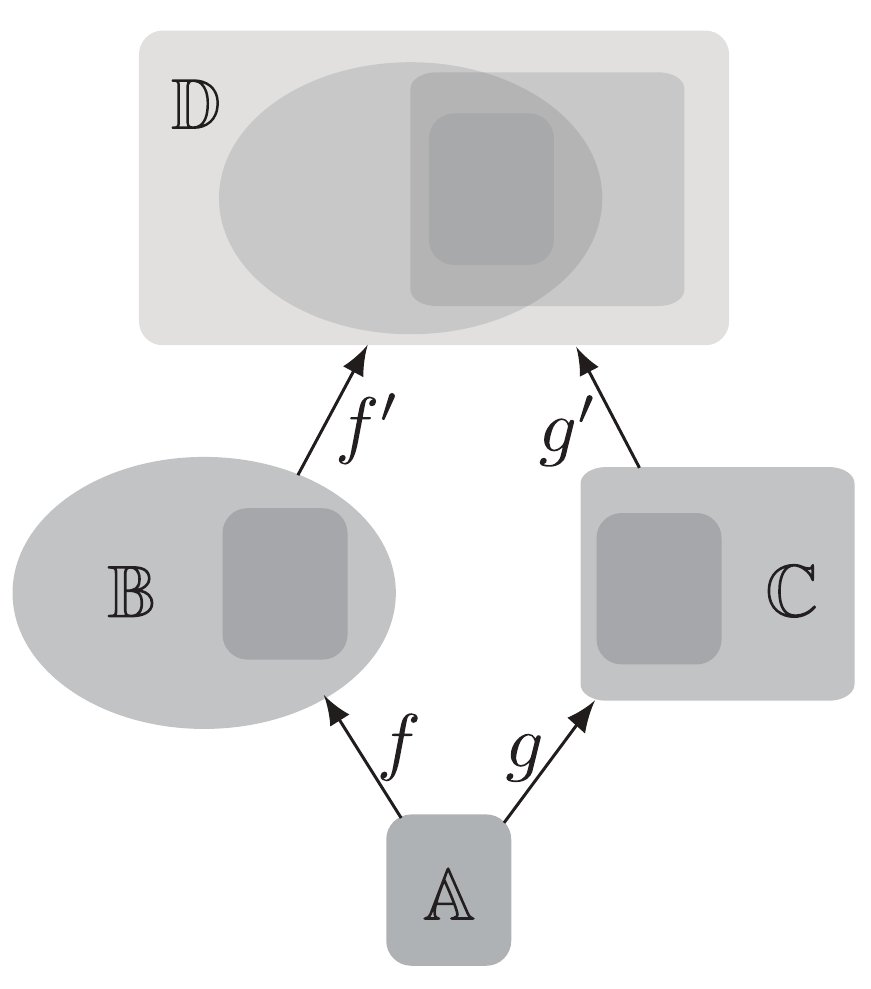}
    \end{center}
}  \caption{Amalgamation of $\str B$ and $\str C$ through $f$ and $g$. All mappings are embeddings.}
  \label{fig:amalg}
\end{figure}

 We say that $\str D$ is a strong amalgam if
$f'(B) \cap g'(C) = (f' \circ f)(A) = (g' \circ g)(A)$, where $A$, $B$
and $C$ denote the domains of $\str A$, $\str B$ and $\str C$,
respectively. We say that $\str D$ is a free amalgam if it is strong
and, additionally, $\str D = \str D[f'(B)] \cup \str D[g'(C)]$.  We
also say that $\str D$ is the union of $\str B$ and $\str C$
amalgamated along $\str A$ through $f$ and $g$ via $f'$ and $g'$.
Note that the free amalgam of $\str B$ and $\str C$ through $f$ and
$g$ is uniquely defined up to isomorphism,
and is isomorphic to the disjoint union of $\str B$
and $\str C$, quotiented by the equivalence relation
identifying $f(x)$ with $g(x)$, for $x$ in $A$.
 We denote this free amalgam $f \cup_{\str A} g$.  When $f$ and
$g$ are  implicit, we denote it $\str B\cup_{\str A}\str C$. We
also say that $\str B$ and $\str C$ are glued~along~$\str A$.

\myparagraph{Classes of structures}
All our structures will have finite or countably infinite
domain. Moreover we assume that all structures have a domain that is a
subset of a common background countable set, say $\Nat$. For a fixed
signature $\Sigma$, a class of structures is a set of structures that
is closed under isomorphisms, i.e.\ if $\str A$ and $\str B$ are
isomorphic structures and $\str A$ belongs to the class, then $\str B$
also belongs to the class.  A class of structures $\cal C$ is closed
under amalgamation if for every two embeddings $f$ and $g$ from
the same structure $\str A$ in $\cal C$ into structures $\str B$ and
$\str C$ in $\cal C$, there exists in $\cal C$ an amalgam of $\str B$ and $\str C$
through $f$ and $g$. A class of finite structures is an
amalgamation class, also called a \fraisse class, if it is
closed under taking induced substructures and amalgamation\oldnote{AA:
  The discussion in this endnote is outdated. Anyway, I never really
  understood why our result holds for $\omega$-categorial classes as
  defined below. AA: We must be more explicit by what we mean by
  amalgamation: free amalgamation? strong amalgamation?  Etc. See
  another endnote below. Sz: I think we are working with arbitrary
  amalgamation.  In fact, I think we even need a weaker version: say a
  class of finite structures is \emph{$\omega$-categorical} if it is
  the set of finitely generated substructures of an
  $\omega$-categorical structure.  Any amalgamation class over finite
  relational signature is an $\omega$-categorical class, since its
  \fraisse limit is $\omega$-categorical. I think $\omega$-categorical
  classes are slightly more general than amalgamation classes over
  finite relational signatures -- they are exactly amalgamation
  classes over finite signatures with relation and function symbols,
  but of uniformly bounded growth (for every $n$, there is a bound on
  the size of $n$-generated structures).  For example, the set of all
  finite-dimensional vector spaces over a fixed finite field is an
  $\omega$-categorical class.}.  
  For example, the class of all finite
  graphs is an amalgamation  class -- in fact, it is closed under free amalgamation -- so is the class of all finite digraphs.
  The class of all finite linear orders is also an amalgamation class, although it is not closed under  free amalgamation.  \fraisse's theorem 
  states that a class is \fraisse if and only if it is the class of finite induced substructures of a homogeneous structure.
  
  For two signatures $\Sigma$ and $\Sigma^+$ with the second
  containing the first, if ${\cal C}$ and ${\cal C}^+$ are classes of
  $\Sigma$-structures and $\Sigma^+$-structures, respectively, then we
  say that ${\cal C}$ is the $\Sigma$-reduct of ${\cal C}^+$ if ${\cal
    C}$ is the class of $\Sigma$-reducts of the structures in ${\cal
    C}^+$.

  \myparagraph{Homogenizable classes} We say that a class of
  $\Sigma$-structures is \emph{homogenizable} if there is a signature
  $\Sigma^+$ extending $\Sigma$, and an amalgamation class ${\cal
    C}^+$ of $\Sigma^+$-structures, such that $\cal C$ is the
  $\Sigma$-reduct of ${\cal C}^+$.  For a class of $\Sigma$-structures
  $\calF$, let $\forbh(\calF)$ denote the class of all finite
  $\Sigma$-structures $\str A$ such that for no $\str F$ in $\calF$
  there is a homomorphism from $\str F$ to $\str A$.  \hubicka and
  \nesetril define a notion of regularity, which we call HN-regularity
  (we omit its technical definition), and prove in Theorem~3.1
  from~\cite{HubickaNesetril2015} that if $\calF$ is a HN-regular
  class of finite connected structures, then $\forbh(\cal F)$ is
  homogenizable. In particular, if $\calF$ is finite, then
  $\forbh(\cal F)$ is homogenizable.

\begin{example}
  Let $\Sigma$ be the signature that consists of one binary predicate
  $\vec E$ and two unary predicates $S$ and $T$.  Let $\str P_n$
  denote a simple directed $\vec E$-path with~$n$ nodes from a unique
  $S$-colored node to a unique $T$-colored node.  The class ${\cal
    F} = \{ \str P_n : n \geq 1 \}$ is HN-regular, 
    and therefore, by~\cite{HubickaNesetril2015}, the class $\forbh(\cal F)$ is homogenizable.
    It consists of digraphs 
    whose nodes are possibly labeled with $S$ or $T$, 
    and there is no directed path from an $S$-labeled node to a $T$-labeled node.
      We show that
  $\forbh(\cal F)$ 
   is homogenizable by a direct
  construction\oldnote{AA: What I wrote ended up being longer and
    messier than I wished; please rewrite if you have a better way to
    phrase it.}.  
    Let $\Sigma^+$ be the extension of $\Sigma$
    by two unary predicates $I$ and $O$.
    Let $\cal C^+$ consist of all $\Sigma^+$-structures $\str A^+$ such that the domain of $\str A^+$
    is partitioned into 
    $I^{\str A^+}$ and $O^{\str A^+}$, and that
    $S^{\str A^+}\subset I^{\str A^+}$,
    $T^{\str A^+}\subset O^{\str A^+}$, and there are no  $\vec E$-edges starting in $I^{\str A^+}$ and ending in $O^{\str A^+}.$ Then $\cal C^+$ is an amalgamation class, as
    it is closed under free amalgamation.
    The class $\forbh (\cal F)$ is the $\Sigma$-reduct
    of $\cal C^+$: a structure $\str A$ in $\forbh(\cal F)$
    expands to a structure $\str A^+$ in $\cal C^+$,
    in which $I^{\str A^+}$ is the set of vertices 
    reachable from $S^{\str A}$ by a directed $\vec E$-path, and $O^{\str A^+}$
    is its complement.
    \qed

\end{example}

\section{Necessary condition for homogenizability} \label{sec:necessary}

Fix a finite relational signature $\Sigma$. Except for the examples,
in this section all structures are over this signature, or over a
signature $\Sigma^+$ that extends $\Sigma$.  Before we state the
necessary condition for homogenizability we need some notation and
terminology.


Let $\cal C$ be a class of finite structures. If $\str A$, $\str L$
and $\str R$ are structures in $\cal C$, and $L : \str A \to \str L$
and $R : \str A \to \str R$ are embeddings such that no amalgam of
$\str L$ and $\str R$ through $L$ and $R$ is in $\cal C$, then we say
that $L : \str A \rightarrow \str L,\; R : \str A \rightarrow \str R$
is a \emph{diagram that witnesses failure} of amalgamation of $\cal
C$. 
We illustrate the definitions with a running example.
\begin{example}[Running example]\label{ex:run2}
    Let $\str F_n$ denote the structure depicted in
    Figure~\ref{fig:structure}, with $n$ vertices in the middle column.
    \begin{figure}[ht]
    \begin{center}
    \includegraphics[height=2in]{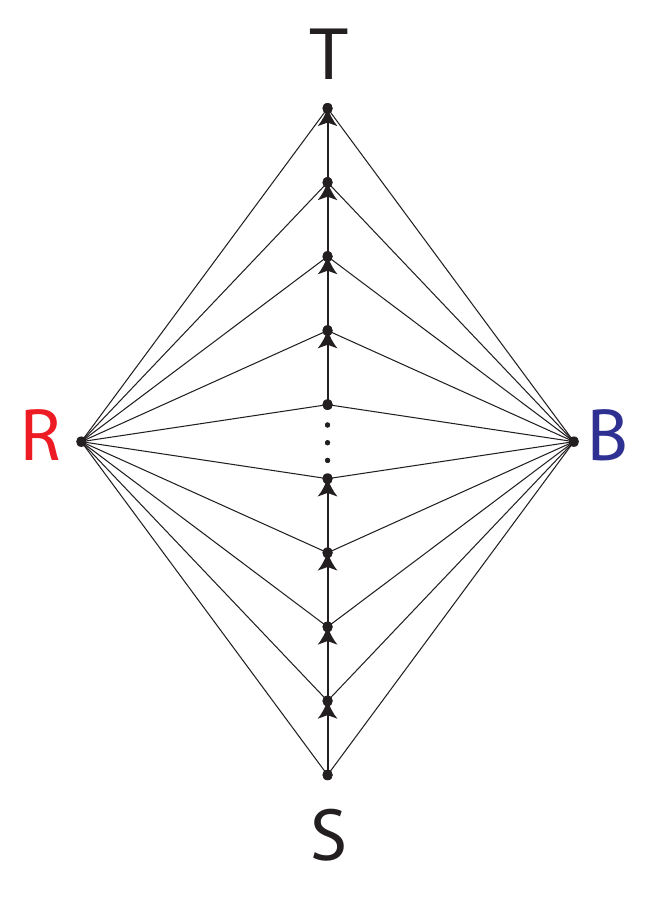}
    \end{center}
    \caption{Forbidden structure $\str F_n$\label{fig:structure}.}
    \end{figure}
    The signature $\Sigma$ of this structure consists of one binary
    predicate~$E$, the undirected edges, one binary predicate~$\vec E$,
    the vertical directed edges, and four unary predicates~$R$ (for \emph{red}), $B$ (for \emph{blue}), $S$
    (for \emph{source}),
    and~$T$ (for \emph{target}), each appearing in the structure exactly once.     
    Observe
    that the colours $S$ and $T$ ensure that $\cal F$ is an antichain in
    the homomorphism pre-order, i.e.\ there are no homomorphisms from
    $\str F_n$ to $\str F_m$ if $n \not= m$.
    Let $\cal C=\forbh(\cal F)$. 
    In the running example, we will demonstrate that the class
    $\cal C$ is not homogenizable.

Choose a large natural number $n$.  Let $\str L$ denote the left part
of the structure $\str F_n$ obtained by removing the blue vertex
(labeled $B$). Symmetrically, let $\str R$ denote the right part of
$\str F_n$ obtained by removing the red vertex (labeled $R$). Let
$\str A$ denote the intersection of $\str L$ and $\str R$, i.e., the
$\vec E$-path with $n$ vertices starting at the $S$-labeled vertex and
ending at the $T$-labeled vertex. Let $L:\str A\to \str L$ and $R:\str
A\to \str R$ be the inclusion mappings. Then any amalgamation of $L$
and $R$ necessarily is a homomorphic image of $\str F_n$. Hence $L :
\str A \to \str L,\; R: \str A \to \str R$ is a diagram that witnesses
failure of amalgamation of $\calC$. \qed
\end{example}

Let $L:\str A\to \str L,\; R:\str A\to \str R$ be a diagram that
witnesses failure of amalgamation of $\cal C$.  For a structure $\str
J$ and a partial mapping $C:{{\str J}\choose{\str A}} \to \{L,R\}$,
let $\str J^C$ be the structure that is obtained by glueing to $\str
J$, for each $\pi$ in $\dom C$, a fresh copy of either $\str L$ or
$\str R$ depending on whether $C(\pi)=L$ or $C(\pi)=R$.  More
formally, $\str J^C$ is defined by induction on the cardinality of the
domain of $C$: if $\dom C=\emptyset$, then $\str J^C=\str J$;
otherwise, if $C=C'\cup\set{\pi\mapsto \sigma}$, where $\pi\in {{\str
    J}\choose{\str A}}$ and $\sigma\in \set{L,R}$, then define $\str
J^{C}={\pi'}\cup_{\str A} \sigma$, where $\pi':\str A\to \str J^{C'}$
is $\pi:\str A\to\str J$ composed with the identity embedding from
$\str J$ to~$\str J^{C'}$.

For a natural number $m$ and a $\Sigma$-structure $\str A$ with domain
$A$, let $\str A\otimes m$ denote the structure with domain $A\times
[m]$ in which the interpretation of a relation $R$ in $\Sigma$ of
arity~$k$ is the set of all tuples
$((a_1,i_1),(a_2,i_2),\ldots,(a_k,i_k))$ where $(a_1,\ldots,a_k)\in
R^{\str A}$ and $i_1,\ldots,i_k\in [m]$. Observe that every function
$f:A\to [m]$ induces an embedding $\pi_f: \str A\to \str A\otimes m$,
defined by $\pi_f(a)=(a,f(a))$. Let $\Ee_{\str A,m}$ denote the set of
all embeddings of the form $\pi_f$ for $f:A\to [m]$. In particular,
$\Ee_{\str A,m}$ is a subset of ${{{\str A\otimes m}\choose{\str A}}}$
containing exactly $m^{|A|}$ embeddings.

A diagram $L:\str A\to \str L,\; R:\str A\to \str R$ is 
\emph{confusing} for $\cal C$ if the following conditions hold:
\begin{enumerate} \itemsep=0pt
	\item it witnesses failure of amalgamation of $\cal C$, and
	\item for every natural number $m$,
  if $\str J=\str A\otimes m$,
	then for every coloring $C:\Ee_{\str A,m}\to\set{L,R}$
	the structure $\str J^C$ belongs to the class $\calC$.	
\end{enumerate}
Its \emph{order} is the cardinality of the domain of $\str A$.

\begin{theorem}\label{thm:necessary}
If $\cal C$ is a homogenizable class of finite
structures, then there exists a natural number $r$
such that every confusing diagram for $\cal C$ has
order at most $r$.
\end{theorem}
This theorem is the main technical result of this paper.  Before we
prove it, we illustrate it by applying it to our running example.

\begin{example} \label{ex:run3}
  Fix natural numbers $m$ and $n$. Let $L:\str A\to \str L$ and
  $R:\str A\to \str R$ be defined as in Example~\ref{ex:run2}.  The
  structure $\str J=\str A\otimes m$ is depicted in
  Figure~\ref{fig:structJ}.
  \begin{figure}[ht]
  \begin{center}
  \includegraphics[height=2in]{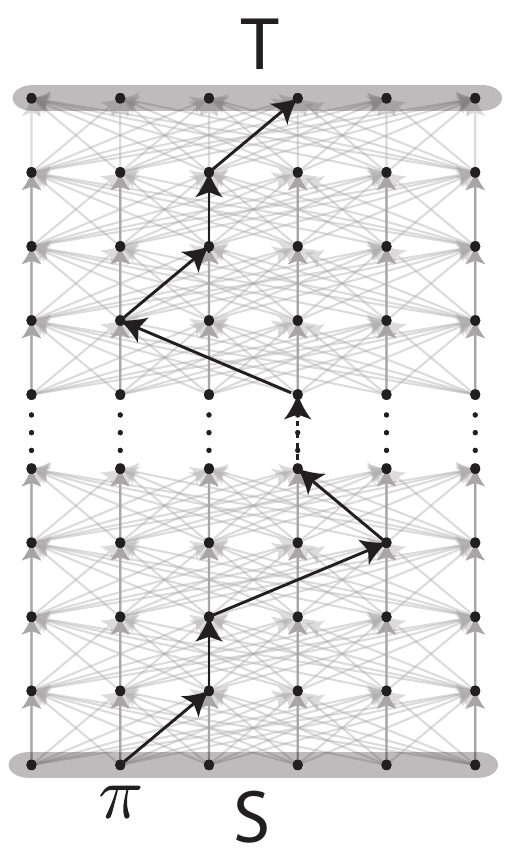}
  \end{center}
  \caption{The structure $\str A\otimes m$\label{fig:structJ},
  with an embedding ${\pi\in \Ee_{\str A,m}}$.
  }
  \end{figure}
Its domain is $[n]\times [m]$, and every element
$(i,j)\in[n]\times[m]$ with $i \leq n-1$ is connected by an $\vec
E$-edge to every element $(i+1,k)\in[n]\times[m]$. The embeddings
$\Ee_{\str A,m}$ correspond to functions $f:[n]\to [m]$.  If
$C:{\Ee_{\str A,m}}\to\set{L,R}$ is a coloring, then $\str J^C$ is
obtained by considering all functions $f:[n]\to [m]$, and connecting
every vertex along the path $\set{(i,f(i)):1\le i\le n}$ to a
fresh vertex which is red if $C(f)=L$, and blue if $C(f)=R$.  Observe
that no structure $\str F$ in $\cal F$ maps homomorphically to $\str
J^C$. Therefore, $\str J^C$ belongs to $\calC=\forbh(\cal F)$.  Since
$m$ is arbitrary, this shows that the diagram $L:\str A\to \str L,\;
R:\str A\to \str R$ is confusing for $\calC$.  Since its order is
$|A|=n$, and $n$ is arbitrary, Theorem~\ref{thm:necessary} implies
that $\calC$ is not a reduct of any amalgamation
class.\qed
\end{example}

Theorem~\ref{thm:necessary} follows easily from 
Lemma~\ref{lem:technical} stated below.

Let $L:\str A\to \str L,\; R:\str A\to \str R$
witness failure of amalgamation of $\calC$.
An $(L,R)$-\emph{confusion} for $\cal C$ is a structure $\str J$ in
$\calC$, together with a set $\cal E\subset {{\str J} \choose {\str
    A}}$, such that $\str J^{C}$ is in $\cal C$ for every coloring
$C:{\cal E}\to \set{L,R}$.  For $\mathcal{E} \subseteq {\str J \choose
  \str A}$ and a natural number $r$ bounded by the cardinality of the
domain of $\str A$, let $\mathcal{E}_r$ denote the set of all
restrictions $\pi|_X$ of $\pi$ in $\mathcal E$, where $X$ ranges over
all $r$-element subsets of the domain of~$\str A$.

\begin{lemma}\label{lem:technical}
Let $r$ and $t$ be natural numbers, and let $\cal C$ be a class of\,
$\Sigma$-structures. There exist numbers $p$ and $q$ (depending on $r$
and $t$ only) such that the following condition implies that $\cal C$
is {not} a reduct of any amalgamation class over a signature with at
most $t$ predicates of arity at most $r$:
\begin{quote}
  there is a diagram $L : \str A \rightarrow \str L,\; R : \str A
  \rightarrow \str R$ that witnesses failure of amalgamation of $\cal
  C$ and of order at least $r$, and there is an $(L,R)$-confusion
  $(\str J,\cal E)$ for $\cal C$ satisfying
  \begin{align}\label{eq:counting}
  |{\mathcal E}| > p
     \cdot |{\mathcal E}_{r}| + q^{{|\str A| \choose
         r}}.    
  \end{align}
\end{quote}
  \end{lemma}
%
%
  Before we prove Lemma~\ref{lem:technical} we show how
  Theorem~\ref{thm:necessary} follows from it.
\begin{proof}[Proof of Theorem~\ref{thm:necessary}]
Suppose that $\calC$ has confusing diagrams of arbitrarily large order.
For every two fixed natural numbers $r$ and $t$, we apply
Lemma~\ref{lem:technical} to conclude that $\cal C$ is not a reduct of
an amalgamation class over a signature with $t$ symbols of arity at
most $r$. Let $p$ and $q$ be as in the statement of the lemma.
Consider a confusing diagram $L:\str A\to \str L,\; R:\str A\to \str
R$ and let $n$ be its order.  Fix a natural number $m$,
and let $\str J= \str A\otimes m$ and $\Ee=\Ee_{\str A,m}$.  Then
$(\str J,\Ee)$ is an $(L,R)$-confusion for ${\cal C}$, by the
definition of confusing diagram,
and 
$|{\mathcal E}| = m^n$ and $|{\mathcal E}_{r}| = m^{r}$.
 Since the order $n$ of the diagram
can be chosen arbitrarily large, we can assume $n>r$. Taking $m$ large
enough, so that $p\cdot m^{n-1}>q^{n\choose r}$ and $m>2p$, we get:
	    \begin{align*}
			p\cdot |{\mathcal E}_{r}| + q^{{|\str A| \choose
			         r}}=
	    p \cdot m^{r} + q^{n \choose r} \le
		p\cdot m^{n-1}+p\cdot m^{n-1}<m^n=|\mathcal E|,
	    \end{align*}
	    which gives condition~\eqref{eq:counting}
	    in Lemma~\ref{lem:technical}.
Since $t$ and $r$ were arbitrary, this proves that
            $\cal C$ is not the reduct of an amalgamation class.
\end{proof}

It remains to prove the lemma.
\begin{proof}[{Proof of Lemma~\ref{lem:technical}}]
  Fix natural numbers $r$ and $t$. In anticipation of the proof, let
  $q$ be the maximum number of atomic types of $(r+1)$-tuples over any
  signature with at most $t$ predicates of arity
  at most $r$, and let $p = \lceil{\log_2(q)}\rceil$.  Suppose that
  $\Cc$ is a class of $\Sigma$-structures as in the lemma, with a
  diagram $L : \str A \rightarrow \str L,\; R : \str A \rightarrow
  \str R$ that witnesses its failure of amalgamation, and an $(L,R)$-confusion $(\str J,\mathcal E)$ satisfying  condition~\eqref{eq:counting} from Lemma~\ref{lem:technical}. 
  

Let $\str B^+$ be a $\Sigma^+$-structure with domain $\str B$ and 
let  $f:A\to B$ be a function from some set $A$ to $B$.
Define the \emph{pullback} $f^*(\str B)$ as the $\Sigma^+$-structure
with universe $A$ in which the interpretation of a relation
symbol $R$ of $\Sigma^+$ of arity $k$ is $f^{-1}(R^{\str B})$, i.e., the inverse image of the interpretation of $R$ in $\str B$ under the mapping $f:A^k\to B^k$. By definition, $f^*(\str B)$ is the unique $\Sigma^+$-structure
on $A$ for which $f$ is a strong homomorphism into $\str B$.

By definition of the structure $\str J^C$,
there is a distinguished embedding of $\str J$ into~$\str J^C$.
Therefore, by composition, any embedding $\pi:\str A\to \str J$ in $\cal E$ defines an embedding of $\str A$ into~$\str J^C$, denoted $\hat\pi:\str A\to \str J^C$. Note that for any expansion $\str J^+$ of $\str J^C$,
the pullback $\hat\pi^*(\str J^+)$ is an expansion of $\str A$, which is isomorphic (via $\hat \pi$) to an induced
substructure of $\str J^+$.

  \begin{claim} \label{postcounting} There is a coloring $C :
    {\mathcal E} \rightarrow \{L,R\}$ such that, for every expansion $\str
    J^+$ of $\str J^C$ over the signature $\Sigma^+$, there are two embeddings $\pi$ and
    $\sigma$ in $\mathcal E$ such that the pullbacks
    $\hat\pi^*(\str J^+)$ and $\hat\sigma^*(\str J^+)$ are equal, but $C(\pi) \not=
    C(\sigma)$.
  \end{claim}

  We show how the claim yields the lemma. Figure~\ref{fig:glue} illustrates the proof.

  Assume that $\Cc$ is the
  class of $\Sigma$-reducts of a class of $\Sigma^+$-structure $\Cc^+$. To reach a contradiction, suppose that $\Cc^+$ is closed
  under induced substructures and amalgamation. Let $C$ be as in the claim. Since $\str
  J^C$ belongs to $\Cc$ by the definition of confusion, there exists an expansion
  $\str J^+$ of $\str J^C$ in $\Cc^+$. Let $\pi$ and $\sigma$ be as in
  the conclusion of the claim, and suppose without loss of 
  generality that $C(\pi)=L$ and $C(\sigma)=R$.
  By the definition of $\str J^C$,
  the embeddings $\pi:\str A\to \str J$ and $L:\str A\to \str L$ induce  embeddings $\hat\pi,\pi', f$,
  such that the diagram to the left below commutes:
  $$\xymatrix{
  &\str A\ar_{\pi}[dl] \ar^L[dr]\ar_{\hat \pi}[dd]\\
  \str J \ar_{\pi'}[rd] 
  &
  &\str L \ar^{f}[ld]\\
  &\str J^C}\qquad\qquad
  \xymatrix{
  &\str A\ar_{\sigma}[dl] \ar^R[dr]\ar_{\hat \sigma}[dd]\\
  \str J \ar_{\sigma'}[rd] 
  &
  &\str R \ar^{g}[ld]\\
  &\str J^C}
  $$  
   Let $\str L^+=f^*(\str J^+)$ be the pullback structure; this structure is an expansion of $\str L$. Moreover,
   $\str L^+$ belongs to the class $\cal C^+$,
   since it is a pullback under an injective mapping, 
   and hence $\str L^+$ is isomorphic to an induced substructure $\str L^+_\pi$ of $\str J^+$, which is in $\cal C^+$.

  Similarly, the embeddings $\sigma:\str A\to \str J$
  and $R:\str A\to \str R$ induce embeddings $\hat\sigma,\sigma',g$ 
  such that the diagram to the right above commutes.
  Let $\str R^+=g^*(\str J^+)$ be the pullback structure, which is an expansion of $\str R$,
  isomorphic to an induced substructure $\str R^+_\sigma$ of $\str J^+$, hence
   belongs to the class $\cal C^+$.

\begin{figure}[ht]
\begin{center}
\includegraphics[height=2.6in]{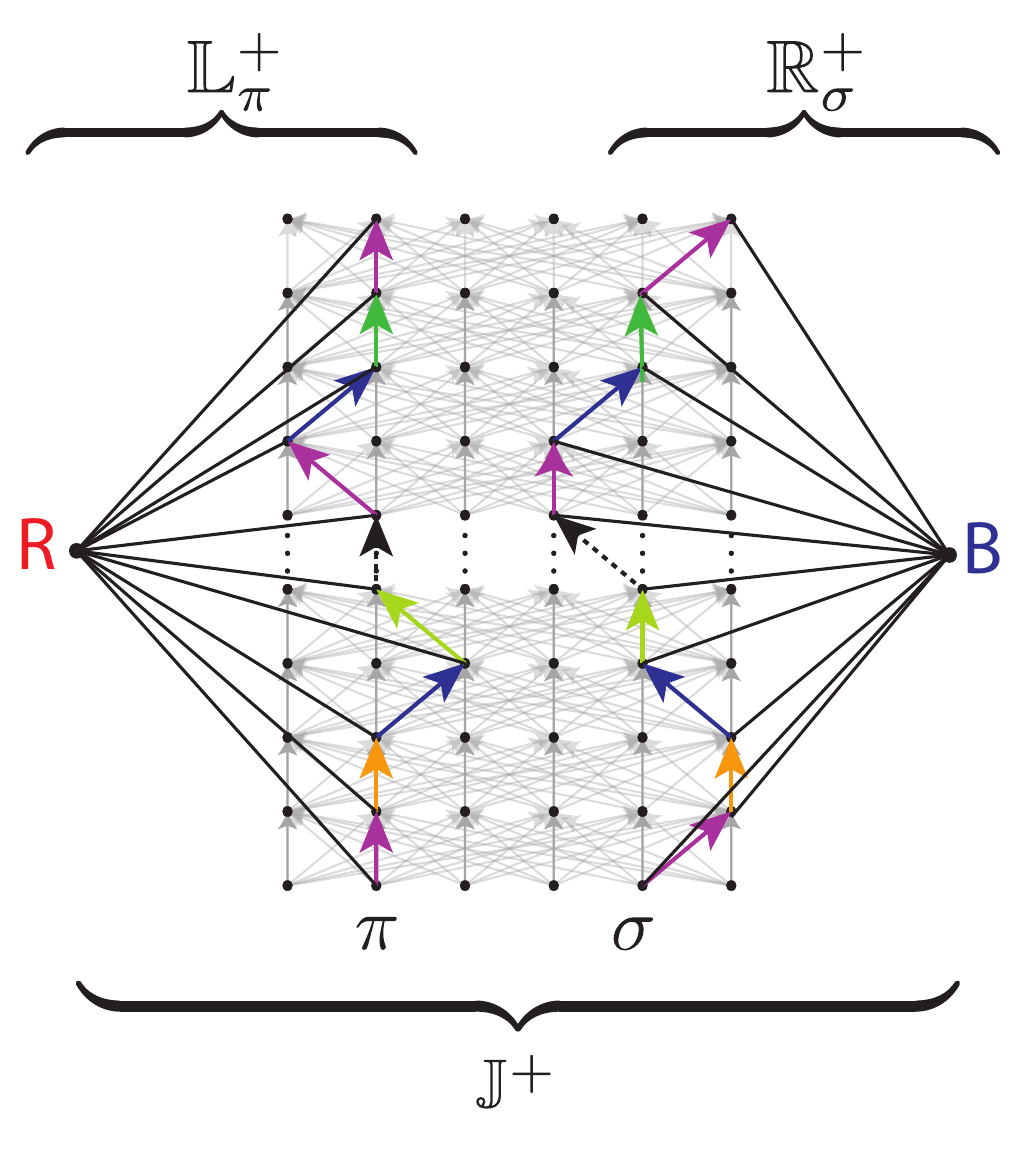}
\end{center}
\caption{    
Two embeddings $\pi,\sigma\in\cal E$
with $C(\pi)=L$ and $C(\sigma)=R$,
 in the context of the running example. The colored arrows depict various predicates
of the stipulated signature $\Sigma^+$ extending~$\Sigma$
(in general, they don't need to be binary).
The fact that the sequences of colors along $\pi$ and along $\sigma$
are the same corresponds to the assumption that 
the pullbacks $\hat\pi^*(\str J^+)$ and $\hat\sigma^*(\str J^+)$
are equal. Therefore, the marked substructures $\str L^+_\pi$ and $\str R^+_\sigma$ of $\str J^+$
(which correspond to $\str L^+$ and $\str R^+$ in the proof
via the mappings $f$ and $g$)
have an isomorphic substructure (isomorphic to 
$\str A^+$ in the proof).
An amalgamation in $\cal C^+$ of $\str L^+_\pi$ and $\str R^+_\sigma$ along this substructure would yield as a $\Sigma$-reduct an amalgamation in $\cal C$ of $\str L$ and $\str R$ along 
$\str A$, a contradiction.
\label{fig:glue}}
\end{figure}

Let $\str A^+$ be the pullback $\hat\pi^*(\str J^+)$,
which, by the claim, is the same as the pullback $\hat\sigma^*(\str J^+)$. 
Note that by commutativity of the  diagram to the left above, the pullback $\hat\pi^*(\str J^+)$ is the same as 
the pullback $L^*(\str L^+)$. Similarly, $\hat\sigma^*(\str J^+)$ is the same as $R^*(\str R^+)$. In other words, 
$L$ is an embedding of $\str A^+$ into $\str L^+$,
and $R$ is an embedding of $\str A^+$ into $\str R^+$.  
%
 Since $\cal C^+$ is closed under amalgamation, there
 exists an amalgamation of the diagram $L:\str A^+\to \str L^+$ and $R:\str A^+\to \str R^+$,
 which consists of a structure $\str U^+$ in $\cal C^+$ and two embeddings $L':\str L^+\to \str U^+$ and $R':\str R^+\to \str U^+$. Taking $\Sigma$-reducts, we obtain an amalgamation in $\cal C$
 of $L:\str A\to \str L$ and $R:\str A\to \str R$.
 But the pair of embeddings $L$ and $R$ were supposed to witness failure of amalgamation
 in $\calC$ -- a contradiction proving that $\Cc^+$ cannot be closed under amalgamation.

  Next we show how to prove Claim~\ref{postcounting} and hence Lemma~\ref{lem:technical}.  Call any  embedding in $\mathcal E$ a  \emph{spot}, and
  any restriction of a spot to an $r$-element subset of the domain of
  $\str A$ a \emph{partial spot}. For each coloring $C$ of
  the spots, and each two spots $\pi$ and $\sigma$, define $\pi
  \approx_C \sigma$ if and only if $C(\pi) = C(\sigma)$.  For each
  coloring $D$ of the partial spots, and each two spots $\pi$ and
  $\sigma$, define $\pi \sim_D \sigma$ if and only if $D(\pi|_X) =
  D(\sigma|_X)$ for every $r$-element subset $X$ of the domain of
  $\str A$. Both are equivalence relations on spots.

  \begin{claim} \label{counting}
    There is a coloring $C$ of the spots using two colors, such that
    for all colorings $D$ of the partial spots using $q$ colors, there
    is a pair of spots $\pi$ and $\sigma$ such that $\pi \sim_D
    \sigma$ but $\pi \not\approx_C \sigma$.
  \end{claim}

\begin{proof}    For this proof, let $n$ be the cardinality of the domain of $\str
    A$ and assume without loss that the domain of $\str A$ is $[n] =
    \{1,\ldots,n\}$. Let $N$ be the number of spots and let $M$ be the
    number of partial spots. With this notation, condition~\eqref{eq:counting}
    reads as follows:
    \begin{equation}
    N > p \cdot M + q^{n \choose r}. \label{eqn:ineq}
    \end{equation}
    Color the spots independently at random with either $L$ or $R$,
    each with probability~$1/2$. Let $C$ be the random variable
    describing this process. In particular, $C$ is a random variable
    taking as values strings of length $N$ over alphabet $\{L,R\}$,
    each with the same probability. Thus the binary entropy $h(C)$ of
    the random variable $C$ is equal to $N$.

    Suppose for contradiction that the opposite of what the claim
    states holds. Then there is a random variable $D$ taking as values
    colorings of the partial spots using $q$ colors such that the
    inclusion $\sim_D \;\subseteq\; \approx_C$ holds with
    probability~$1$. The relation $\sim_D$ has at most $q^{n \choose
      r}$ equivalence classes; for each fixed spot $\pi$, there are at
    most $q$ choices of colors for each of the $n \choose r$
    restrictions $\pi|_X$ to $r$-element subsets $X \subseteq [n]$,
    and any two spots sharing these choices are equivalent. In
    particular, there is a random variable $E$ taking as values
    strings of length $q^{n \choose r}$ over alphabet $\{L,R\}$ such
    that $E$ and $D$ determine $C$. That is, $h(C \;|\; E,D) = 0$, or
    equivalently,
    $$
    h(C,E,D) = h(E,D).
    $$
    We will show that this is impossible by proving that
    $$
    h(E,D) < N = h(C) \leq h(C,E,D).
    $$
    Indeed, $D$ is determined by $n \choose r$ random variables $\{
    D_X : X \subseteq [n], |X| = r \}$, where the random variable
    $D_X$ takes as values the colorings of the restrictions of the
    spots to the subset~$X$.  If $M_X$ denotes the number of such
    restrictions, the random variable $D_X$ takes values in
    $[q]^{M_X}$, and therefore
    $$
    h(D_X) \leq \log(q^{M_X}) = \log_2(q) \cdot M_X \leq p \cdot M_X.
    $$
    Noting that $M$ is the sum of $M_X$ as $X$ ranges over all
    $r$-element subsets of $[n]$, it follows that
    $$
    h(D) \leq \sum_{X} h(D_X) \leq \sum_{X} p \cdot M_X = p \cdot
   \sum_{X} M_X = p \cdot M.
    $$
    Moreover $h(E) \leq q^{n \choose r}$ since $E$ takes as values
    strings of length $q^{n \choose r}$ over alphabet $\{L,R\}$. Hence
    $$
    h(E,D) \leq h(E) + h(D) \leq q^{n \choose r} + p \cdot M.
    $$
    However~\eqref{eqn:ineq} states that this quantity is strictly
    smaller than $N$, as required.
\end{proof}

Finally we use Claim~\ref{counting} to prove Claim~\ref{postcounting}.
Let $C$ be the coloring of Claim~\ref{counting} with the two colors
interpreted as the embeddings $L : \str A \rightarrow \str L$
and $R : \str A \rightarrow \str R$. For each expansion $\str J^+$ of $\str
J^C$, let $D$ be the coloring of partial spots defined as follows: for
each spot $\pi$ and each $r$-element subset $X = \{i_1 < \ldots <
i_r\}$ of the domain of $\str A$, let $D(\pi|_X)$ be the atomic type
of $(\pi(i_1),\ldots,\pi(i_r))$ in $\str J^+$. This is a coloring of
partial spots using at most $q$ colors. By Claim~\ref{counting}, there
is a pair of spots $\pi$ and $\sigma$ such that $\pi \sim_D \sigma$
but $\pi \not\approx_C \sigma$. From $\pi \sim_D \sigma$ and the fact
that $r$ is at least as large as the maximum arity of any new predicate
in $\Sigma^+$, it follows that the pullbacks $\hat\pi^*(\str J^+)$ and $\hat\sigma^*(\str J^+)$ are equal. On the other hand,
from $\pi \not\approx_C \sigma$ we get $C(\pi) \not= C(\sigma)$ by
definition. This proves Claim~\ref{postcounting} and Lemma~\ref{lem:technical}.
\end{proof}

\section{Classes of consistent structures} \label{sec:consistent}

In this section we work out the first of our two examples of
non-homogenizable classes. We start by defining some basic notions
from the theory of constraint satisfaction problems as described, for
example, in Chapter~6 of the monograph \cite{Graedeletalbook}. Recall
that, for a structure $\str T$, we write $\csp(\str T)$ for the class
of all finite structures $\str I$ over the same signature as $\str T$
for which there is a homomorphism from $\str I$ to $\str T$.  The
$\str I$'s are called instances, the $\str T$'s are called templates.

\subsection{Local consistency}

Let $\Sigma$ be a relational signature, let $\str A$ and $\str B$ be
$\Sigma$-structures, and let $k$ and $l$ be integers such that $1\le
k\le l$. A $(k,l)$-consistent family on $\str A$ and $\str B$ is a
non-empty family $\cal F$ of partial homomorphisms from $\str A$ to
$\str B$, such that the following three conditions hold for each $f$
in $\cal F$:
\begin{enumerate} \itemsep=0pt
  \item $|\dom f|\le l$,
  \item if $X$ is a subset of $A$, then $f|_X$ is in $\cal F$,
\item if $|\dom f|\le k$ and $X$ is a subset of $A$ such that $\dom
  f\subset X$ and $|X|\le l$, then there exists $g$ in $\cal F$ such
  that $\dom g=X$ and $f\subset g$.
\end{enumerate}  
If there is a $(k,l)$-consistent family on $\str A$ and $\str B$, then
we say that $\str A$ is $(k,l)$-consistent with respect to $\str B$.
Note for later use that the class of structures that are
$(k,l)$-consistent with respect to $\str B$ is closed under inverse
homomorphisms: if there is a homomorphism from $\str A'$ to $\str A$,
and $\str A$ is $(k,l)$-consistent with respect to $\str B$, then
$\str A'$ is also $(k,l)$-consistent with respect to $\str B$. To see
this, it suffices to compose the homomorphism from $\str A'$ to $\str
A$ with each partial homomorphism in the $(k,l)$-consistent family for
$\str A$ to get a $(k,l)$-consistent family for $\str A'$.

We describe the special case of $(2,3)$-consistency in terms of a
pebble game.
The game is played between spoiler and duplicator, each having three
pebbles, numbered $1$, $2$ and $3$.  Spoiler can place his pebbles on
the nodes of $\str A$, while duplicator can place his pebbles on the
nodes of $\str B$. They can also keep the pebbles in their pockets, in
which they have all pebbles at the beginning of the game.  The game
proceeds in rounds as follows.
In each round, spoiler places some of the pebbles from his pocket on
the nodes of $\str A$ and duplicator replies by placing his
corresponding pebbles on the nodes of $\str B$.  If the partial
mapping defined by the pebble placement is not a partial homomorphism
from $\str A$ to $\str B$, then duplicator loses.  Otherwise, spoiler
puts back some of the pebbles into his pocket, and duplicator removes
the corresponding pebbles, and the game continues to the next round.
It is not hard to see that $\str A$ is $(2,3)$-consistent with respect
to $\str B$ if and only if duplicator can avoid losing forever.


\subsection{Systems of linear equations over $\mathbb{F}_2$}

We define a finite template $\str T_2$ that can be used to represent
the solvability of systems of linear equations over the 2-element
field.  Let us note that our definition of the template $\str T_2$
will not be the standard one as it can be found, for example, in the
original Feder-Vardi paper \cite{FederVardi1998}. The main difference
is that we want to have a signature of smallest possible arity, in
this case two. We achieve this by letting $\str T_2$ be the natural
encoding of the standard template as its \emph{incidence} structure.
Concretely, $\str T_2$ is defined as follows. Its domain is $D\cup R$,
where
  \begin{align*}
  D &= \{0,1\}, \\ 
  R &= \set{(x,y,z) \in D^3 : x+y+z=0 \;\text{mod}\; 2}.
  \end{align*}
  The elements of $D$ are called values, and those of $R$ are called
  triples.  The signature $\Sigma$ includes three partial functions
  $\pi_1$, $\pi_2$, and $\pi_3$ that map triples in $R$ to values in
  $D$, and four unary relations \emph{value}, \emph{triple}, $C_0$ and
  $C_1$.  Formally, in order to have a relational structure, $\str
  T_2$ has binary relations that correspond to the graphs of the
  partial functions $\pi_1$, $\pi_2$ and $\pi_3$.  The interpretations
  of the symbols in $\str T_2$ are as follows:
  \begin{enumerate} \itemsep=0pt
  \item $\pi_1$, $\pi_2$ and $\pi_3$ map
    $(x,y,z)$ in $R$ to $x$, $y$ and $z$, respectively,
  \item \emph{value} holds of all elements in $D$, 
  \item \emph{triple} holds of all elements in $R$, 
  \item $C_0$ holds of $0$ in $D$, and
  \item $C_1$ holds of $1$ in $D$.
  \end{enumerate}
The purpose of \emph{triple} is to encode equations of the type
$x+y+z=0 \text{ mod } 2$, and the purposes of $C_0$ and $C_1$ are to
encode equations of the type $x=0$ and $x=1$, respectively.  Note that
even though the language does not allow writing more complicated
equations, such as $x + y + z = 1 \text{ mod } 2$ or $w + x
+ y + z = 0 \text{ mod } 2$, such equations can be simulated in
the language of $\str T_2$ with the help of auxiliary variables.

\begin{theorem}\label{thm:cons}
  The class of all finite structures that are $(2,3)$-consistent with
  respect to $\str T_2$ is not homogenizable.
\end{theorem}

  \newcommand{\leaves}{\text{leaves}} 
  \newcommand{\lft}{\mathit{left}}
  \newcommand{\rgt}{\mathit{right}}
  \newcommand{\ftr}{\mathit{father}}

\begin{proof}
  In the following, we fix the template $\str T = \str T_2$, and when
  we refer to consistency, we mean $(2,3)$-consistency with respect to
  $\str T$. Finite structures on the signature of $\str T$ are called
  instances. Homomorphisms $f : \str I \rightarrow \str T$ from an
  instance $\str I$ to $\str T$ are called solutions. By $\calC$
  denote the class of consistent instances. Observe that, as noted
  earlier, the class of consistent instances is closed under inverse
  homomorphisms.

  The plan is to apply Theorem~\ref{thm:necessary} to $\calC$, and for
  that we need to find a confusing diagram $L:\str A\to\str L,R:\str A\to \str R$ with arbitrarily large $\str A$.

  Let $n \geq 8$ be an exact power of two\oldnote{AA: I decided to go
    with complete binary trees to simplify.}. Let $t$ be a rooted,
  ordered tree with $n$ leaves at depth $\log_2(n)$; in particular, no
  node at depth~$2$ is a leaf, and no node at depth $\log_2(n) - 1$ is
  a root.  Let $\str I$ be the instance obtained from $t$, with
  elements of two types: \emph{nodes}, which correspond to the nodes
  of $t$, and \emph{triples}, which correspond to triples
  $(v,v_0,v_1)$, where $v$ is an internal node in $t$, and $v_0$ and
  $v_1$ are its left and right sons, respectively. Nodes are labeled
  by the unary predicate \emph{value} and triples are labeled by the
  unary predicate \emph{triple}.  We say that the triple $(v,v_0,v_1)$
  is the triple \emph{below} node $v$, and is \emph{adjacent} to, or
  \emph{contains} $v$, $v_0$, and $v_1$. For each such triple, we
  declare:
\begin{align*}
& \ftr(v, v_0,v_1) = \pi_1(v,v_0,v_1) = v,\\
& \lft(v, v_0,v_1) = \pi_2(v,v_0,v_1) = v_0,\\
& \rgt(v, v_0,v_1) = \pi_3(v,v_0,v_1) = v_1.
\end{align*}
We call a structure of this kind simply a tree. Since we will work
with $\Sigma$-structures that are made of trees, for the sake of
intuition from now on we use the names $\ftr$, $\lft$, and $\rgt$ in
place of $\pi_1$, $\pi_2$, and $\pi_3$.  If $i$ is a value in $D$,
then the \emph{$i$-marking} of $\str I$ is the $\Sigma$-structure
$m_i(\str I)$ obtained from $\str I$ by marking the root by the
predicate $C_i$. Observe that in any solution $v:m_i(\str I)\to \str
T$ of $m_i(\str I)$, the sum of the values of the leaves is equal to
$i$ modulo~$2$. Conversely, any mapping from the leaves of $\str I$ to
$\str T$ such that the sum of the values of the leaves is equal to $i$
modulo~$2$ extends uniquely to a solution $v:m_i(\str I)\to \str T$.

The structures $\str L$ and $\str R$ are the markings $m_0(\str I)$
and $m_1(\str I)$ of the tree $\str I$, respectively.  The structure
$\str A$ is the substructure of $\str I$ induced by the leaves of the
tree. Note that $\str A$ consists of $n$ isolated points, labeled by
the unary relation \emph{value}. The unary relations \emph{triple},
$C_0$ and $C_1$, as well as the binary relations $\pi_1$, $\pi_2$, and
$\pi_3$, are empty in $\str A$. Note that $\str L$ and $\str R$ share
$\str A$ as an induced substructure.  Let $L : \str A \rightarrow \str
L$ and $R : \str A \rightarrow \str R$ be the corresponding embeddings.

  \begin{lemma}\label{lem:incons}
    The free amalgam of $\str L$ and $\str R$ through $L$ and $R$ is
    inconsistent.
  \end{lemma}
\begin{proof}
  When spoiler has only two pebbles on the board, we allow him to
  perform a move we call a \emph{slide}, in which he moves one pebble
  from a node $v$ to a triple adjacent to it, or from a triple to a
  node belonging to this triple.  Duplicator has to respond
  accordingly: if spoiler slides a pebble from a node $v$ to a triple
  $t$ containing $v$ on the $i$-th coordinate, then duplicator must
  move his corresponding pebble from a value $v$ in $D$ to a triple in
  $R$ containing $v$ on the $i$-th coordinate. Symmetrically, in the
  case when spoiler slides his pebble from a triple to the node in the
  $i$-th coordinate, duplicator must move his corresponding pebble
  from the corresponding triple to the value on its $i$-th coordinate.
  The slide moves can be simulated in the original game, using a third
  pebble.
  
  Denote the two (overlapping) trees $\str I_L$ and $\str I_R$,
  respectively; they have common leaves in the free amalgam $\str
  L\cup_{\str A}\str R$.  Here is the strategy for spoiler; it
  consists of several steps.  In the beginning of the $n$-th step,
  spoiler has two pebbles placed on corresponding nodes $a$ and $b$ of
  $\str I_L$ and $\str I_R$, at depth $n-1$ of the tree. In
  particular, in the beginning of the first step, two pebbles are
  placed on the roots of $\str I_L$ and $\str I_R$, respectively.  For
  a node $v$ on which spoiler has his pebble, denote by $r(v)$ the
  value of the corresponding pebble of duplicator.  The invariant is
  that $r(a)\neq r(b)$. This invariant is clearly satisfied in the
  beginning of the first step, since $\str I_L$ has its root labeled
  with $C_0$ and $\str I_R$ has its root labeled with~$C_1$.
    
  In the $n$-th step, spoiler slides his pebble from node $a$ to the
  triple $a'$ below $a$ in $\str I_L$, and then slides his pebble from
  node $b$ to the triple $b'$ below $b$ in $\str I_R$.  Duplicator's
  responses have to satisfy $r(\ftr(a'))=r(a)$ and $r(\ftr(b'))=r(b)$.
  In particular, $r(\ftr(a'))\neq r(\ftr(b'))$, by the invariant. It
  follows that $\lft(r(a'))+\rgt(r(a'))\neq\lft(r(b'))+\rgt(r(b'))$,
  so either $\lft(r(a'))\neq \lft(r(b'))$ or $\rgt(r(a'))\neq
  \rgt(r(b'))$ (or both). Since the cases are symmetric, suppose
  without loss of generality that the first case occurs. Then spoiler
  slides the pebble from $a'$ to $\lft(a')$ and then slides the pebble
  from $b'$ to $\rgt(b')$, and continues the game from these two nodes
  playing the role of $a$ and $b$.  The invariant is satisfied.
    
  Since in each step the depth of $a$ increases by $1$, at some point,
  $a$ must be a leaf of $\str I_L$, and $b$ is the corresponding leaf
  in $\str I_R$.  But then $a$ and $b$ are the same element in $\str
  L\cup_{\str A}\str R$, and by the invariant $r(a)\neq r(b)$. In
  other words, spoiler has two pebbles placed at the same node of
  $\str L\cup_{\str A} \str R$, but the corresponding pebbles of
  duplicator are not placed on the same element of $\str T$. So
  duplicator loses.
\end{proof}

\begin{lemma}\label{lem:incons1}
Every amalgam of $\str L$ and $\str R$ through $L$ and $R$
is inconsistent.
\end{lemma}

\begin{proof}
This follows at once from the previous lemma and the fact that $\cal
C$ is closed under inverse homomorphisms. Indeed, the free amalgam
$\str L \cup_{\str A} \str R$ through $L$ and $R$ maps homomorphically
to any amalgam of $\str L$ and $\str R$ through $L$ and $R$.
\end{proof}


Let $m$ be a natural number, and let $\str J=\str A\otimes m$ and $\Ee=\Ee_{\str A,m}$.
%
%

  \begin{lemma}\label{lem:cons} 
    For every coloring $C: \mathcal{E} \to\set{L,R}$,
  the structure $\str J^C$ is consistent.
  \end{lemma}
  
  \begin{proof}
We modify the game, by giving more power to spoiler.
We show that even in this game, duplicator wins.
In the modified game, the pebbles of spoiler can be placed only on
triples of $\str J^C$, and the pebbles of duplicator can be 
placed only on triples of $\str T$. If the pebbles of spoiler are placed on
triples $a_1,\ldots,a_k$, with $k\le 3$, then duplicator must have his
corresponding pebbles placed on triples $t_1,\ldots,t_k$ in $\str T$,
so that the following conditions are satisfied:
\begin{itemize}
\item Whenever $a_i$ is a triple containing a node with unary
  predicate $C_j$ on some coordinate, then the same coordinate of
  $t_i$ is equal to $j$.
\item Whenever $a_i$ and $a_j$ agree on some coordinate, then $t_i$
  and $t_j$ also agree on the same coordinate.
\end{itemize}
We show how spoiler can copy a strategy
which is winning in the original game 
to win in the modified game.

\begin{claim}
  If spoiler has a winning strategy in the original game, then he also
  has a winning strategy in the modified game.
\end{claim}
\begin{proof}
  Suppose that in the original game spoiler places a pebble on a node
  $v$.  We copy this move in the modified game by placing a pebble on
  any triple containing $v$ on some coordinate, say, the $i$-th
  coordinate, and await the response of duplicator.  If in the
  modified game duplicator places his corresponding pebble on a triple
  $t$ in $\str T$, then we pretend that the duplicator in the original
  game places his pebble on the $i$-th coordinate of $t$, and the game
  continues. At some point, duplicator loses in the original
  game. This means that one of two cases occurred in the original
  game:
\begin{itemize} \itemsep=0pt
\item Spoiler has placed a pebble on a node with unary predicate $j$
  and duplicator replied by placing his corresponding pebble on a
  value $j'$ with $j'\neq j$.
\item One pebble of spoiler is placed on a node $v$ and another pebble
  of spoiler is placed on a triple $t$ containing $v$ on the $i$-th
  coordinate, and the corresponding pebbles of duplicator are placed
  on a value $r(t)$ and a triple $r(t)$ that, however, do not satisfy
  the condition that the $i$-th coordinate of $r(v)$ equals $r(t)$.
\end{itemize}
Since duplicator is only copying his strategy from the modified game,
it must be the case that duplicator must have lost as well in the
modified game. In particular, if spoiler wins in the original game,
then he wins in the modified game.
\end{proof}

We show a winning strategy for duplicator in the modified game on
$\str J^C$.  By the claim above, this means that duplicator also has a
winning strategy in the original game.

The arena $\str J^C$ on which spoiler places his pebbles is a union of
trees of the form $\str I$ glued along the leaves. Therefore, it is
meaningful to talk about roots, children (or sons), brothers, and
leaves, and parents in the case of nodes from trees which are not
roots nor leaves (leaves have very many parents). Every triple in
$\str J^C$ is of the form $(v,v_0,v_1)$, where $v$ is an internal node
of some tree, and $v_0$ and $v_1$ are its left and right son.

Call two tree nodes $v$ and $w$ in $\str J^C$ \emph{congruent}, and write
$v\cong w$, if the following conditions hold:
\begin{itemize} \itemsep=0pt
  \item The nodes correspond to the same node in $\str I$,
  \item The leaves below $v$ coincide with the leaves below $w$. 
\end{itemize}
We lift this notion to triples: two triples $(v,v_0,v_1)$ and
$(w,w_0,w_1)$ are congruent, also written $(v,v_0,v_1) \cong
(w,w_0,w_1)$, if $v\cong w$, $v_0\cong w_0$, and $v_1\cong
w_1$. Observe that two distinct roots in $\str J^C$ are not congruent,
since by construction, not all their leaves are identified.

%

During the game, let $a_1,\ldots,a_k$, with $k\le 3$, denote the
triples on which the pebbles of spoiler are placed.  
Let $X$ denote the set of nodes that are congruent to some component
of some pebbled triple, and let $X'$ denote the union of $X$ with the
roots\oldnote{AA: Note the subtle change in the definition of $X$ and
  $X'$; earlier $X$ was the set of triples that are congruent to a
  pebbled triple, and $X'$ were the set of components of triples in
  $X$. I think the new definition is more appropriate since $X'$ is
  now closed under congruence.}.  We say that a function $f:X'\to D$ is
\emph{nice} if it satisfies the following conditions.
\begin{enumerate} \itemsep=0pt
\item For every triple $(x,y,z)$ in $\str J^C$, if $x,y,z\in X'$, then
  $f(x)+f(y)+f(z)=0$.
\item For every root $r$, if $r$ is marked with unary predicate $C_i$,
  then $f(r)=i$.
\item Whenever $x,y\in X'$ are congruent, then $f(x)=f(y)$.
\end{enumerate}
We show that duplicator has a strategy which satisfies the following
invariant at each moment of the game:
\begin{quote}
  There is a nice function $f:X'\to D$ such that for each pebble of
  spoiler occupying a triple $(x,y,z)$, duplicator's corresponding
  pebble occupies the triple $(f(x),f(y),f(z))$.
\end{quote}

At the beginning of the game, the invariant is satisfied: since $X'$
consists only of roots, we can define $f(x)=i$ for a root $x$ with
unary predicate $C_i$, yielding a nice function -- the last condition
of nicety holds since no two distinct roots are congruent.

Suppose that at some moment during the game there is a function $f$
satisfying the above conditions, and spoiler performs a move.  If in
this move he removes a pebble from some triple, then duplicator
responds by removing the corresponding pebble from $\str T$, and it is
easy to see that the restriction of $f$ to the resulting set $X'$
satisfies the above conditions.

Suppose now that spoiler makes his move by placing a new pebble on the
board. In particular, before the move he had $k\le 2$ pebbles on
triples $a_1,\ldots,a_k$, and a new pebble is placed on the triple
$a_{k+1}$, which we denote $c$ for simplicity.  Below, unless
indicated, when we speak about $X$, $X'$, or $f$, we refer to their
values just before spoiler placed the new pebble on $c$.  The case
that $c$ is a triple $(v,v_0,v_1)$ with $v,v_0,v_1\in X'$ is trivial:
duplicator just responds with $(f(v),f(v_0),f(v_1))$. This response is
not loosing thanks to the invariant and the first two conditions of
the nicety of $f$.  Moreover, the values of $X$ and $X'$ after
duplicator's response are unmodified, so the same function $f$ can be
used in the invariant. From now on we assume that at least one of
the coordinates of $c$ is not in $X'$.

Note that after spoiler's move, the new $X'$ includes the congruence
classes of the three components of $c$. We say that a triple is
\emph{completed} by spoiler's move if not all three components of the
triple are in $X'$ before spoiler's move, but the addition of these
congruence classes to $X'$ makes all three components of the triple
belong to the new $X'$. In particular, $c$ and its congruents are
completed by spoiler's move. The new $f$ after spoiler's move will be
defined to extend the old $f$ by assigning values to the components of
$c$ and its congruents in such a way that the conditions of nicety are
satisfied for the new $X'$. We need to distinguish several cases:

 Case 1: $c$ is a triple $(v,v_0,v_1)$ in which $v_0$ and $v_1$
  are not leaves, and $v$ is already in $X'$. Let $v_{00}$ and
  $v_{01}$ be the left and right sons of $v_0$, and let $v_{10}$ and
  $v_{11}$ be those of $v_1$.  We need the following claim:

\begin{claim} There exist values $i$, $i_0$, and $i_1$ in $D$ such that
  \begin{enumerate} \itemsep=0pt
    \item $i_0=f(v_0)$ if $v_0$ belongs to $X'$,
    \item $i_1=f(v_1)$ if $v_1$ belongs to $X'$,
    \item $i + i_0 + i_1 = 0$, where $i = f(v)$,
    \item $i_0+f(v_{00})+f(v_{01})=0$ if $v_{00}$ and $v_{01}$
     belong to $X'$,
  \item $i_1+f(v_{10})+f(v_{11})=0$ if $v_{10}$ and $v_{11}$
     belong to $X'$.
  \end{enumerate} 
  \end{claim}
  \begin{proof}
    Since not all three $v$, $v_0$ and $v_1$ are in $X'$ but $v$ is in
    $X'$, at most one among $v_0$ and $v_1$ is in $X'$. It follows
    that not all four $v_{00}$, $v_{01}$, $v_{10}$, and $v_{11}$ can
    be in $X'$. To argue for this, note that at most two pebbles occupy
    at most two triples $t_1$ and $t_2$ before spoiler's move, but it
    cannot be the case that $t_1 \cong (v_0,v_{00},v_{01})$ and $t_2
    \cong (v_1,v_{10},v_{11})$ if not both $v_0$ and $v_1$ are in
    $X'$.  Moreover, for the same reason, if both $v_{00}$ and
    $v_{01}$ are in $X'$, then $v_1$ is not in $X'$, and if both
    $v_{10}$ and $v_{11}$ are in $X'$, then $v_0$ is not in $X'$. We
    use this to choose $i_0$ and $i_1$ by cases.

    Case (i): both $v_{00}$ and $v_{01}$ are in $X'$. First choose
    $i_0$ to satisfy condition~\emph{4} and then choose $i_1$ to
    satisfy condition~\emph{3}. Note that in case condition~\emph{1}
    also applies, then the only choice of $i_0$ that makes
    condition~\emph{4} hold is guaranteed to satisfy
    condition~\emph{1} too by the first condition of nicety of
    $f$. Note also that in this case conditions~\emph{2} and~\emph{5}
    do not apply.

    Case (ii): both $v_{10}$ and $v_{11}$ are in $X'$. First choose
    $i_1$ to satisfy condition~\emph{5} and then choose $i_0$ to
    satisfy condition~\emph{3}. Again, note that in case
    condition~\emph{2} also applies, then the only choice of $i_1$
    that makes condition~\emph{5} hold is guaranteed to satisfy
    condition~\emph{2} too by the first condition of nicety of $f$. Note
    also that in this case conditions~\emph{1} and~\emph{4} do not
    apply.

    Case (iii): otherwise. In this case the only conditions that can
    apply are~\emph{1},~\emph{2}, and~\emph{3}, and among~\emph{1}
    and~\emph{2} at most one can apply. In case~\emph{1} applies and
    $v_0$ is in $X'$, first choose $i_0$ to satisfy condition~\emph{1}
    and then choose $i_1$ to satisfy condition~\emph{3}. In
    case~\emph{2} applies and $v_1$ is in $X'$, first choose $i_1$ to
    satisfy condition~\emph{2} and then choose $i_0$ to satisfy
    condition~\emph{3}.  \end{proof}

Case 2: $c$ is a triple $(v,v_0,v_1)$ in which $v_0$ and $v_1$
  are not leaves, and $v$ is not yet in $X'$. Let $v_{00}$ and $v_{01}$ be
  the left and right sons of $v_0$, and let $v_{10}$ and $v_{11}$ be
  those of $v_1$.  Since $v$ is not in $X'$, it is not a root. Let
  then $w$ be the sibbling of $v$, and let $u$ be their parent. We
  need the following claim:

\begin{claim} There exist values $i$, $i_0$, and $i_1$ in $D$ such
that
\begin{enumerate} \itemsep=0pt
\item $i_0 = f(v_0)$ if $v_0$ belongs to $X'$,
\item $i_1 = f(v_1)$ if $v_1$ belongs to $X'$,
\item $i + i_0 + i_1 = 0$,
\item $f(u) + f(w) + i = 0$ if $w$ and $u$ belong to $X'$,
\item $i_0 + f(v_{00}) + f(v_{01}) = 0$ if $v_{00}$ and $v_{01}$ belong to $X'$,
\item $i_1 + f(v_{10}) + f(v_{11}) = 0$ if $v_{10}$ and $v_{11}$ belong to $X'$.
\end{enumerate}
\end{claim}

\begin{proof}
  If both $v_0$ and $v_1$ are in $X'$, we argue that $w$ and $u$ are
  not in $X'$. To see this, note that at most two pebbles occupy at
  most two triples before spoiler's move. But if both $v_0$ and $v_1$
  are in $X'$, then these triples must contain nodes that are
  congruent to $v_0$ and $v_1$, and be different and hence different
  from any triple that contains a node congruent to $u$ or $w$, since
  all triples that contain both $v_0$ and $v_1$ are congruent to
  $c$. Thus, in case both $v_0$ and $v_1$ are in $X'$, we choose $i_0
  = f(v_0)$ and $i_1 = f(v_1)$, and $i$ to satisfy
  condition~\emph{3}. Conditions~\emph{1},~\emph{2} and~\emph{3} are
  then true by construction, condition~\emph{4} does not apply, and
  conditions~\emph{5} and~\emph{6} hold because $f$ is nice with
  respect to $X'$.

  Assume then that not both $v_0$ and $v_1$ are in $X'$. In such a
  case we argue that not all four $v_{00}$, $v_{01}$, $v_{10}$, and
  $v_{11}$ can be in $X'$. To see this, note again that at most two
  pebbles occupy at most two triples $t_1$ and $t_2$, and it cannot be
  that $t_1 \cong (v_0,v_{00},v_{01})$ and $t_2 \cong
  (v_1,v_{10},v_{11})$ if not both $v_0$ and $v_1$ are in $X'$.
  Moreover, for the same reason, if both $v_{00}$ and $v_{01}$ are in
  $X'$, then $v_1$ is not in $X'$, and if both $v_{10}$ and $v_{11}$
  are in $X'$, then $v_0$ is not in $X'$. We use this to choose $i_0$
  and $i_1$ by cases. In all cases we first choose $i$ to satisfy
  condition~\emph{4}.

  Case (i): both $v_{00}$ and $v_{01}$ are in $X'$. First choose $i_0$
  to satisfy condition~\emph{5} and then choose $i_1$ to satisfy
  condition~\emph{3}. Note that in case condition~\emph{1} also
  applies, then the only choice of $i_0$ that makes condition~\emph{5}
  hold is guaranteed to satisfy condition~\emph{1} too by the first
  condition of nicety of $f$. Note also that in this case
  conditions~\emph{2} and~\emph{6} do not apply.

    Case (ii): both $v_{10}$ and $v_{11}$ are in $X'$. First choose
    $i_1$ to satisfy condition~\emph{6} and then choose $i_0$ to
    satisfy condition~\emph{3}. Again, note that in case
    condition~\emph{2} also applies, then the only choice of $i_1$
    that makes condition~\emph{6} hold is guaranteed to satisfy
    condition~\emph{2} too by the first condition of nicety of $f$. Note
    also that in this case conditions~\emph{1} and~\emph{5} do not
    apply.

    Case (iii): otherwise. In this case the only conditions that can
    apply are~\emph{1},~\emph{2}, and~\emph{3} (and~\emph{4}), and
    among~\emph{1} and~\emph{2} at most one can apply. In
    case~\emph{1} applies and $v_0$ is in $X'$, first choose $i_0$ to
    satisfy condition~\emph{1} and then choose $i_1$ to satisfy
    condition~\emph{3}. In case~\emph{2} applies and $v_1$ is in $X'$,
    first choose $i_1$ to satisfy condition~\emph{2} and then choose
    $i_0$ to satisfy condition~\emph{3}.  \end{proof}

 Case 3 (and last): $c$ is a triple $(v,v_0,v_1)$ in which $v_0$
  and $v_1$ are leaves. Since $v$ is not a root, let $w$ be its
  sibling, and let $u$ be their parent.

\begin{claim} There exist values $i$, $i_0$, and $i_1$ in $D$ such
that
\begin{enumerate} \itemsep=0pt
\item $i = f(v)$ if $v$ belongs to $X'$,
\item $i_0 = f(v_0)$ if $v_0$ belongs to $X'$,
\item $i_1 = f(v_1)$ if $v_1$ belongs to $X'$,
\item $i + i_0 + i_1 = 0$,
\item $f(u) + f(w) + i = 0$ if $w$ and $u$ belong to $X'$,
\end{enumerate}
\end{claim}

\begin{proof} As in the previous case, if both $v_0$ and $v_1$ are in
  $X'$, then $w$ and $u$ are not in $X'$, but the argument to show why
  this is the case is slightly different. First note that if both
  $v_0$ and $v_1$ are in $X'$ then $v$ is not in $X'$ because not all
  three components of $c$ are in $X'$ by assumption. Second, at most
  two pebbles occupy at most two triples before spoiler's move. If
  both $v_0$ and $v_1$ are in $X'$, then these triples must contain
  $v_0$ and $v_1$, which are congruent only to themselves, and be
  different and hence different from any triple that contains a node
  congruent to $u$ or $w$, since all the triples that contains both
  $v_0$ and $v_1$ are congruent to $c$. Thus, in case both $v_0$ and
  $v_1$ are in $X'$, we choose $i_0 = f(v_0)$ and $i_1 = f(v_1)$, and
  $i$ to satisfy condition~\emph{4}.  Conditions~\emph{1} and~\emph{5}
  just do not apply.

  Assume then that not both $v_0$ and $v_1$ are in $X'$. In such a
  case, first choose $i$ to satisfy condition~\emph{5}. Note that if
  condition~\emph{1} also applies, then the unique choice that
  satisfies~\emph{5} also satisfies~\emph{1} by the first condition of
  the nicety of $f$. Once $i$ is chosen, choose either $i_0$ or $i_1$
  to satisfy whichever condition among~\emph{2} or~\emph{3} applies,
  if any, and then choose the other to satisfy condition~\emph{4}.
\end{proof}

This completes the cases analysis over $c$. Now, fix $i$, $i_0$, and
$i_1$ as in the claim in whichever of the three cases applies. We
claim that $f$ can be extended to a function $g$ that is defined on
$v$, $v_0$, and $v_1$ so that $g(v) = i$, $g(v_0)=i_0$, and
$g(v_1)=i_1$, and that is nice with respect to the new $X'$.  Indeed,
let $Y$, $Y_0$, and $Y_1$ denote the sets of nodes that are congruent
to $v$, $v_1$, and $v_1$, respectively. We define the extension $g$ of
$f$ by setting $g(x) = i$ for all $x \in Y$, $g(x)=i_0$ for all $x\in
Y_0$, and $g(x)=i_1$ for all $x\in Y_1$. By the choices of $i$, $i_1$
and $i_2$ in the claims, and the third condition of nicety of $f$,
this is well defined for those $x$ on which $f$ was already
defined. Note that the domain of $g$ is precisely the value of $X'$
after spoiler's move. Let us argue that $g$ is nice with respect to
this new $X'$.

  First we note that on all triples that are congruent to
  $(v,v_0,v_1)$, its three components get the same three values
  $(g(v),g(v_0),g(v_1))$. This shows that $g$ satisfies the third
  condition of nicety with respect to the new $X'$. The second
  condition of nicety is also satisfied since $g$ extends $f$ and
  $f$ was nice with respect to the old $X'$, which contained all roots
  already. Finally, in order to argue that $g$ satisfies the first
  condition of nicety we need to argue which triples are completed
  by spoiler's move. The triple $c$ and its congruents are definitely
  completed and, for these, the condition $i + i_0 + i_1 = 0$ from the
  claims guarantees the first condition of nicety. The addition of
  $v$, $v_0$, and $v_1$ to $X'$ can complete the triples $(u,v,w)$,
  $(v_0,v_{00},v_{01})$, and $(v_1,v_{10},v_{11})$, when they exist,
  and their congruents, but no other triples. And for these, the
  conditions of the claims guarantee that the choices of $i$, $i_0$,
  and $i_1$ satisfy the first condition of nicety.
\end{proof}

Lemma~\ref{lem:incons1} and Lemma~\ref{lem:cons}
show that the diagram $L:\str A\to \str L,R:\str A\to \str R$ is confusing for the class of consistent structures.
Since $\str A$ can be taken arbitrarily large, 
Theorem~\ref{thm:cons} follows immediately from
Theorem~\ref{thm:necessary}.
\end{proof}

\subsection{Other finite Abelian groups}

The template $\str T_2$ for systems of equations over the 2-element
field can be generalized to all finite Abelian groups. Let $G$ be a
finite Abelian group; we write $+$ for the group operation and $0$ for
its neutral element. Let $\str T_G$ be the structure with domain $D
\cup R$, where
  \begin{align*}
  D &= G, \\ 
  R &= \set{(x,y,z) \in D^3 : x+y+z=0 }.
  \end{align*}
The elements of $D$ are called values, and those of $R$ are called
triples. As in $\str T_2$, the signature of $\str T_G$ has three
binary relations $\pi_1$, $\pi_2$, and $\pi_3$, two unary relations
\emph{value} and \emph{triple}, and one unary relation $C_a$ for each
value $a$ in $D$.  The interpretations of all relation symbols are as
in $\str T$; in particular, the unary relation symbol $C_a$ is
interpreted by the singleton set $\{a\}$. It is straightforward to
check that $\str T_G$ can be used to encode arbitrary systems of
equations over $G$. As in the 2-element field case, equations more
complex than the basic $x + y + z = 0$ or $x = a$ can be encoded with
the help of auxiliary variables.

\begin{theorem}\label{thm:consgeneral}
  If $G$ is a finite Abelian group with at least two elements, then
  the class of all finite structures that are $(2,3)$-consistent with
  respect to $\str T_G$ is not homogenizable.
\end{theorem}

\begin{proof} The proof
of Theorem~\ref{thm:cons} does not rely in any way on the fact that
the group is addition mod~2, except for it being Abelian and having at
least two different values in it. \end{proof}

It is known that, for any non-trivial finite Abelian group, the
constraint satisfaction problem of the template $\str T_G$ has
\emph{unbounded width}, i.e.\ for every two natural numbers $k$ and
$l$ there exist instances $\str I$ that do not have homomorphisms to
$\str T_G$, but are nonetheless $(k,l)$-consistent with respect to
$\str T_G$. We also say that $\str T_G$ does not have $(k,l)$-width
for any $k$ and $l$.  This was proved by Feder and Vardi
\cite{FederVardi1998} for the standard template for linear equations
mod 2, and later alternative proofs generalize quite well to the
template $\str T_G$ (see, for instance,
\cite{AtseriasBulatovDawar2009}). Moreover, the solution to the
Bounded-Width Conjecture of Barto and Kozik \cite{BartoKozik2014} implies
that all cases of templates of unbounded width are explained by the
unbounded width of some $\str T_G$. Technically:
 
\begin{theorem}[\cite{BartoKozik2014}, see also Theorem 4.1 in 
\cite{Barto2015}]\label{thm:bartokozik}
\oldnote{The closest published statement that I found is Theorem 4.1
  in the survey by Barto in the Bulletin of Symbolic Logic.}  Let
$\str T$ be a core finite relational structure. If $\str T$ does not
have bounded width, then $\str T$ pp-interprets $\str T_G$ for some
non-trivial finite Abelian group $G$. Moreover, if the signature of $\str T$
has maximum arity at most two, then the conclusion holds even if
$\str T$ does not have $(2,3)$-width.
\end{theorem}

Thus, the templates $\str T_G$ are in a strict formal sense the
canonical templates of unbounded width.  Theorem~\ref{thm:consgeneral}
states that, for all such templates, their class of locally consistent
instances is non-homogenizable.  Interestingly, the converse to this
is also true in a strong sense: for all templates that \emph{do have}
bounded width, their class of locally consistent instances \emph{is}
homogenizable.  This follows quite directly from the fact that, for
every finite template $\str T$, the class of instances $\str I$ that
have a homomorphism to $\str T$ is homogenized by expanding them by
all their homomorphisms to $\str T$.  When these two observations are
put together, we get that, up to the relation of pp-interpretability
between templates, which is known to preserve the property of having
bounded width, the templates that have bounded width are distinguished
from those that do not by the fact that their classes of locally
consistent instances are homogenizable.  It seems plausible that
our Theorem~\ref{thm:consgeneral} could be adapted to show that
\emph{all} templates of unbounded width give themselves a
non-homogenizable class of locally-consistent instances, without the
need to resort to pp-interpretability, but this remains open.


\section{Classes defined by forbidden homomorphisms}
\label{sec:forbidden}
The positive result of \hubicka and
\nesetril~\cite{HubickaNesetril2015} shows that if $\calF$ is an
HN-regular class of finite connected structures, then the class
$\forbh(\calF)$ is a reduct of an amalgamation class. HN-regularity is
a notion reminiscent of the notion of regularity of word languages or
of tree languages. Indeed, in the case of structures of treewidth one,
HN-regularity and regularity in the sense of tree automata both
correspond to MSO-definability.  In this section we give an example of
a non-homogenizable class of finite structures that is of the form
$\forbh({\cal G})$, where ${\cal G}$ is an MSO-definable class of
connected finite structures of treewidth two.
It follows from the result of \hubicka and \nesetril that this is optimal.
 Recall that the
treewidth of a finite structure is defined as the treewidth of its
Gaifman graph, and pathwidth is a restriction of treewidth (for
definitions see, for example, \cite{CourcelleEngelfriet2012}).


\subsection{Pathwidth three}

Consider the class $\cal F$ from the running example in
Section~\ref{sec:necessary}. This was shown non-homogenizable in
Example~\ref{ex:run3}.  The class $\cal F$ can be defined by an MSO
sentence, which expresses that there are exactly four colored points,
which are colored~$R$, $B$, $S$, and~$T$, respectively, and the rest
of points form a directed simple $\vec E$-path from $S$ to $T$ with
all vertices along the path connected by an undirected $E$-edge to
both $R$ and $B$. Moreover, each structure $\str F$ in $\cal F$ is
connected and has pathwidth three: take the path-decomposition of the
$\vec E$-path with bags of size~$2$ and add both red and blue vertices
to each bag.  This gives a path-decomposition with bags of size~$4$,
so its pathwidth is~$3$ (thanks to the $-1$ in the definition of
treewidth/pathwidth).
%
Now we show how to modify the class $\calF$ to obtain a class of structures of treewidth~two.

\subsection{Treewidth two}\label{sec:tw2}
Consider a rooted, directed binary tree, in which every
node is either a leaf, or an inner node with
two sons, in which case it has a directed $\vec E_0$-edge
to its left son and a directed $\vec E_1$-edge to its right son. Color its root red, by labeling it with the unary predicate $R$, and create an extra blue vertex, with unary predicate $B$, connected to all the leaves of the tree by an undirected $E$-edge.
An example of such a structure, obtained from a full binary 
tree of depth $4$, is  depicted in
Figure~\ref{fig:structure2}. Let $\cal G$ denote the class of all structures obtained in this way.
\begin{figure}[ht]
\begin{center}
\includegraphics[height=2in]{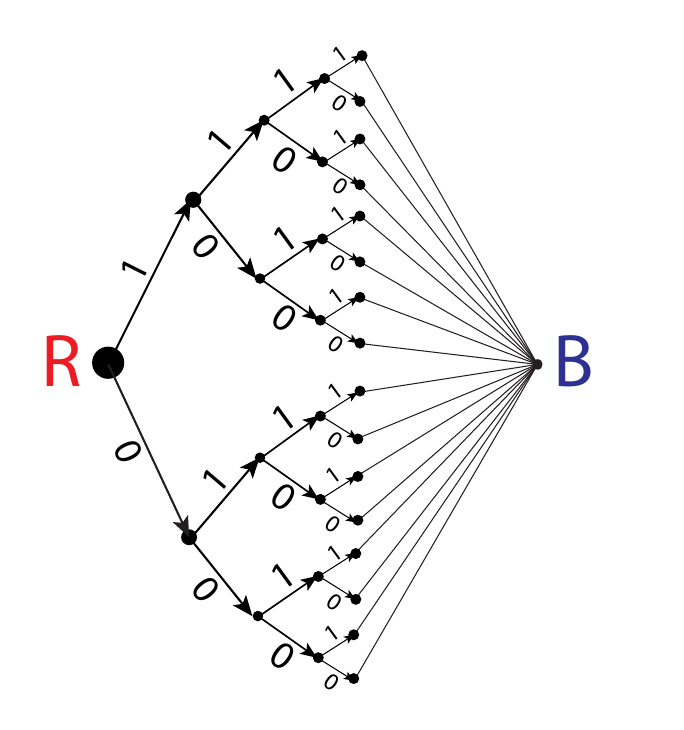}
\end{center}
\caption{Forbidden structure $\str G$, with 16 leaves. \label{fig:structure2}}
\end{figure}
The signature $\Sigma$ of these structures consists of three binary
predicates~$E$, $\vec E_0$, and $\vec E_1$ for the edges, and two unary predicates~$R$
and~$B$, each appearing in the structure exactly once as indicated.
It is straightforward to check that $\cal G$ is MSO-definable on the
class of all finite structures.
 Moreover, each structure $\str G$ in $\cal G$ is connected
and has treewidth two: just take a tree-decomposition of the
binary tree and add the blue point to all its bags.
\begin{claim}\label{claim:antichain}
The class $\cal G$ forms an antichain in the homomorphism preorder.  
\end{claim}
\begin{proof}
  Suppose that $h:\str G_1\to\str G_2$ is a homomorphism of two
  structures in $\cal G$.  Then $h$ must map the root of $\str G_1$ to
  the root of $\str G_2$ (since only the root is colored red), and
  must map the leaves of $\str G_1$ to the leaves of $\str G_2$ (since
  only the leaves are adjacent to a blue node). Finally, a vertex $v$
  in $\str G_1$ reached from the root by a path with labels
  $i_1i_2\ldots i_k\in\set{0,1}^*$ must be mapped to the unique vertex
  $w$ of $\str G_2$ reached from the root by the path obtained by
  reading the same
  labels. Hence, the mapping $f$ is injective. Since no inner node of
  the tree can be mapped to a leaf, $f$ must also be surjective. It
  follows that $h$ is an isomorphism.
\end{proof}

%
%

\begin{proposition} \label{prop:treewidthtwo}
  The class $\forbh({\cal G})$ is not homogenizable.
\end{proposition}
\begin{proof}
  We apply Theorem~\ref{thm:necessary}.
  To this end, choose an arbitrary structure $\str G\in\cal G$, and consider the diagram $L:\str A\to \str L,
  R:\str A\to \str R$ defined as follows.
 $\str L$ is the left part of the structure $\str G$,
 obtained by removing the blue vertex (labeled $B$),
 $\str R$ is the right part of the structure $\str G$,
 obtained by keeping the blue vertex and the nodes adjacent to it, $\str A$ is the intersection of $\str L$ and $\str R$, i.e., the substructure of $\str G$ induced by the leaves of the underlying binary tree. Let $L:\str A\to\str L$ and $R:\str A\to \str R$ be the two inclusions.
 It is clear that every amalgamation of $L,R$ 
 must contain a homomorphic image of $\str G$, so
 $L,R$ witnesses failure of amalgamation of $\forbh(\cal G)$. Let $m$ be an arbitrary number,
 $\str J=\str A\otimes m$, $\Ee=\Ee_{\str A,m}$,
and $C:\Ee\to\set{L,R}$ be any coloring.
\begin{claim}\label{claim:confusing}
The structure $\str J^C$
does not contain a homomorphic image of any structure in $\cal G$. 
  \end{claim}
  
\begin{proof}
  Assume that $h:\str G'\to \str J^C$ is a homomorphism and $\str G'\in \cal G$. Then the root $v$ of $\str G'$ is mapped to some red vertex $h(v)$ in $\str J^C$,
  let $\pi:\str A\to \str J$ be the embedding corresponding to the vertex $h(v)$. In particular, $C(\pi)=L$.
  Let $f$ be the embedding of $\str L$ into $\str J^C$ 
  induced by the embedding $\pi:\str A\to\str J$ and $L:\str A\to \str L$. 
   By the same argument as in the proof
  of Claim~\ref{claim:antichain}, $h$ must map the nodes of $\str G'$ injectively into the structure $f(\str L)$.
  Moreover, $h$  maps the leaves of $\str G'$ into
  elements of $\str J\subseteq\str J^C$ (since only those
  vertices may be adjacent to a blue node). As in the proof 
  of Claim~\ref{claim:antichain}, it follows that $h$
  is a bijection from $\str G'$ without the blue node to the image of $f$.
  In particular, the leaves in $\str G'$ are mapped
  bijectively to the image of $\pi$. But since $C(\pi)=L$,
  it is impossible that all the nodes in the image of $\pi$
  are adjacent to a common blue vertex.
\end{proof}

Hence, the diagram $L,R$ is confusing for $\calC$.
Since $\str G$ can be chosen so that $\str A$ is arbitrarily large, the conclusion follows from Theorem~\ref{thm:necessary}.  
\end{proof}

\subsection{Optimality}

We argued already that the classes $\cal F$ from Example~\ref{ex:run2}
and $\cal G$ from Section~\ref{sec:tw2} are MSO-definable. Therefore,
the set of colored paths that represent the path-decompositions of
the structures in $\cal F$ is regular in the automata-theoretic sense,
and the set of colored trees that represent the tree-decompositions
of the structures in $\cal G$ is regular in the
tree-automata-theoretic sense (see \cite{CourcelleEngelfriet2012}).
It is interesting to check why $\cal F$ and $\cal G$ are \emph{not}
regular classes of structures in the sense of Definition~2.3 of
\hubicka-\nesetril \cite{HubickaNesetril2015}. By
Example~\ref{ex:run3} and Proposition~\ref{prop:treewidthtwo} we know
that $\cal F$ and $\cal G$ cannot be HN-regular as otherwise
$\forbh({\cal F})$ and $\forbh({\cal G})$ would be homogenizable by
Theorem~3.1 in~\cite{HubickaNesetril2015}.

In order to check that a class is not HN-regular it suffices to
identify minimal g-separating g-cuts of unbounded sizes in its
structures. For $\cal F$, note that the set of all vertices in the
$\vec E$-path is a minimal g-separating g-cut in $\str F_k$, and its
size is $k$ and hence unbounded. For $\cal G$, the set of all leaves
in the binary tree in any structure $\str G$ in $\cal G$ is a minimal
g-separating g-cut, and its size is also unbounded since all trees are
represented in $\cal G$.
%

To close this section we note that every MSO-definable class of finite connected structures
of treewidth one \emph{is} HN-regular. This follows from the fact
noted earlier that, for colored trees, HN-regularity, tree-automata
regularity, and MSO-definability are equivalent. In particular, by
Theorem~3.1 in \cite{HubickaNesetril2015}, every class of the form
$\forbh({\cal F})$, where $\cal F$ is an MSO-definable class of
connected finite structures of treewidth at most one, is
homogenizable.

\section{Conclusion}
We study homogenizability -- a combinatorial notion useful in computer
science (see
e.g.~\cite{Bodirsky2008},\cite{BojanczykKlinLasota2011},\cite{BojanczykSegoufinTorunczyk2013}).
Our main contribution is a necessary condition for homogenizability of
a class of finite structures.  We apply it to prove
non-homogenizability of a class related to constraint satisfaction
problems, consisting of locally consistent instances with respect to
the template for linear equations over a finite Abelian group, and a
class defined by forbidding homomorphisms from an MSO-definable class
of structures of treewidth two, which is tight by the positive result
of~\cite{HubickaNesetril2015}.

Our original motivation for studying the homogenizability of classes
of CSP instances came from an approach, first outlined in
\cite{AtseriasWeyer2009}, to characterize the finite templates that
are solvable by any sound consistency algorithm. This was applied in
\cite{AtseriasWeyer2009} to get (yet) a(nother) criterion to decide
solvability by the arc-consistency algorithm, and it was asked if this
kind of technology could also work for $(k,l)$-consistency. While our
negative results of Section~\ref{sec:consistent} rule out the direct
applicability of this method to $(k,l)$-consistency, perhaps an
indirect method could still work for these or other consistency
algorithms. For example, perhaps one could first decide if the class
of consistent instances with respect to the template $\str T$ is
homogenizable, and apply the method from \cite{AtseriasWeyer2009} only
in the relevant case that it is.

This raises the very interesting question of deciding whether a
finitely presented class of finite structures is homogenizable, and
for that we need conditions that are both necessary and
sufficient. Significant steps in that direction were taken in the
\hubicka-\nesetril paper \cite{HubickaNesetril2015} for classes of the
form $\forbh({\cal F})$. Perhaps special cases that are still enough
for the method to work are easier. Can we characterize the classes of
the form $\forbh({\cal F})$, with MSO-definable $\cal F$ say, that are
reducts of classes that are closed under induced substructures and
free amalgamation?  Such classes we call freely homogenizable.

\paragraph{Acknowledgments} First author partially funded
            by the European Research Council (ERC) under the European
            Union's Horizon 2020 research and innovation programme
            grant agreement ERC-2014-CoG 648276 AUTAR) and Ministerio
            de Econom\'{\i}a y Competitividad grant
            TIN2013-48031-C4-1-P TASSAT-2. Part of this work was done
            while the second author was visiting Universitat
            Polit\`ecnica de Catalunya funded by the European Research
            Council (ERC) under the European Union's Horizon 2020
            research and innovation programme grant agreement
            ERC-2014-CoG 648276 AUTAR).

\bibliographystyle{plain} \bibliography{biblio.bib}

\end{document}